\documentclass{amsart}     

\usepackage{virginialake}
\vllineartrue
\usepackage{amsmath}
\usepackage{amsthm}
\usepackage{hyperref}
\usepackage{soul}

\usepackage{amsaddr}

\newcommand{\dfn}{:=}
\renewcommand{\phi}{\varphi}
\renewcommand{\emptyset}{\varnothing}

\newcommand{\eg}{e.g.}

\newcommand{\IH}{\mathit{IH}}

\renewcommand{\vec}[1]{\mathbf{#1}}

\newtheorem{theorem}{Theorem}
\newtheorem{observation}[theorem]{Observation}
\newtheorem{question}[theorem]{Question}
\newtheorem{corollary}[theorem]{Corollary}
\newtheorem{proposition}[theorem]{Proposition}
\newtheorem{lemma}[theorem]{Lemma}
\newtheorem*{question*}{Question}

\theoremstyle{definition}
\newtheorem{definition}[theorem]{Definition}
\newtheorem{example}[theorem]{Example}
\newtheorem{notation}[theorem]{Notation}

\newcommand{\cfal}{\mathsf{f}}
\newcommand{\ctru}{\mathsf{t}}
\newcommand{\cnot}[1]{\overline{#1}}
\newcommand{\cor}{\vee}
\newcommand{\cand}{\wedge}

\newcommand{\lneg}[1]{\cnot{#1}}
\newcommand{\limp}{ \multimap}
\newcommand{\plimp}{\limp^+}

\newcommand{\seqar}{\vdash}

\renewcommand{\models}{\vDash}
\newcommand{\notmodels}{\nvDash}

\newcommand{\MALL}{\mathsf{MALL}}
\newcommand{\aMALL}{\MALL\mathsf{w}}
\newcommand{\focMALL}{\mathsf{F}\MALL}
\newcommand{\afocMALL}{\mathsf{F}\aMALL}
\newcommand{\Var}{\mathsf{Var}}
\newcommand{\Nat}{\mathbb{N}}
\newcommand{\cpl}{\mathsf{CPL}}
\newcommand{\qcpl}{\cpl2}

\newcommand{\sat}{\mathsf{SAT}}

\newcommand{\Prov}[1]{#1\text{-}\mathsf{Prov}}

\newcommand{\dec}{D}
\newcommand{\codec}{\bar{\dec}}
\newcommand{\rel}{R}
\newcommand{\corel}{\bar{\rel}}

\newcommand{\id}{\mathit{id}}
\newcommand{\wkid}{\mathit{wid}}
\newcommand{\wk}{\mathit{w}}

\newcommand{\lefrul}[1]{#1_l}
\newcommand{\rigrul}[1]{#1_r}

\newcommand{\taut}{\mathit{tr}}
\newcommand{\cut}{\mathit{cut}}

\newcommand{\negtrans}[1]{#1^{-}}
\newcommand{\postrans}[1]{#1^+}
\newcommand{\qltrans}[1]{[#1]}
\newcommand{\lqtrans}[1]{\langle #1 \rangle}

\newcommand{\ndcomp}{\lceil \sigma\rceil}
\newcommand{\condcomp}{\lceil\pi \rceil}

\newcommand{\ptime}{\mathbf{P}}
\newcommand{\nptime}{\mathbf{NP}}
\newcommand{\conptime}{\mathbf{coNP}}
\newcommand{\pspace}{\mathbf{PSPACE}}
\newcommand{\pth}{\mathbf{PH}}

\makeatletter
\def\@settitle{\begin{center}%
		\baselineskip14\p@\relax
		\bfseries
		\uppercasenonmath\@title
		\@title
		\ifx\@subtitle\@empty\else
		\\[2ex]\uppercasenonmath\@subtitle
		\footnotesize\mdseries\@subtitle
		\fi
	\end{center}%
}
\def\subtitle#1{\gdef\@subtitle{#1}}
\def\@subtitle{}
\makeatother

\begin{document}

\title{From QBFs to MALL and back via focussing }
\subtitle{Fragments of multiplicative additive linear logic for\\ each level of the polynomial hierarchy}

\author{Anupam Das}
\address{University of Birmingham}
\email{a.das@bham.ac.uk}
\thanks{The author was supported by a Marie Sk\l{}odowska-Curie fellowship, \href{http://cordis.europa.eu/project/rcn/209401_en.html}{ERC project 753431}, while conducting this research.}
\date{\today}

\maketitle

\begin{abstract}
In this work we investigate how to extract alternating time bounds from `focussed' proof systems.
Our main result is the obtention of fragments of $\aMALL$ ($\MALL$ with weakening) complete for each level of the polynomial hierarchy. 
In one direction we encode QBF satisfiability and in the other we encode focussed proof search, and we show that the composition of the two encodings preserves quantifier alternation, yielding the required result. 
By carefully composing with well-known embeddings of $\aMALL$ into $\MALL$, we obtain a similar delineation of $\MALL$ formulas, again carving out fragments complete for each level of the polynomial hierarchy.
This refines the well-known results that both $\aMALL$ and $\MALL$ are $\pspace$-complete.

A key insight is that we have to refine the usual presentation of focussing to account for \emph{deterministic} computations in proof search, which correspond to invertible rules that do not branch. 
This is so that we may more faithfully associate phases of focussed proof search to their alternating time complexity.
This presentation seems to uncover further dualities at the level of proof search than usual presentations, so could be of further proof theoretic interest in its own right.
\end{abstract}

	\section{Introduction and motivation}
%
%
%
%
%
Proof search is one of the most general ways of deciding formulas of expressive logics, both automatically and interactively.
In particular,
proof systems can often be found to yield optimal decision algorithms, in terms of asymptotic complexity.
To this end, we now know how to extract bounds for proof search in terms of various properties of the proof system at hand.
For instance we may establish:
\begin{itemize}
	\item nondeterministic time bounds via \emph{proof complexity}, \eg\  \cite{CookReck:79:The-Rela:mf,Krajicek:1996:BAP:225488,Cook:2010:LFP:1734064};
	\item (non)deterministic space bounds via the \emph{depth} of proofs or search spaces, and \emph{loop-checking}, \eg\ \cite{buss2003depth,heuerding1996efficient,ono1998proof}; 
	\item deterministic or co-nondeterministic time bounds via systems of \emph{invertible rules}, see \eg\ \cite{TroeSchw:96:Basic-Pr:mi,negri2008structural}.
\end{itemize}


\noindent
However, despite considerable progress in the field, there still remains a gap between the obtention of (co-)nondeterministic time bounds, such as $\nptime$ or $\conptime$, and space bounds such as $\pspace$.
Phrased differently, while we have many logics we know to be $\pspace$-complete (intuitionistic propositional logic, various modal logics, etc.), we have very little understanding of their fragments corresponding to subclasses of $\pspace$.

An alternative view of space complexity is in terms of \emph{alternating time complexity}, where a Turing machine may have both existential (i.e.\ nondeterminstic) and universal (i.e.\ co-nondeterministic) branching states.
In this way $\pspace $ is known to be equivalent to {alternating polynomial time} \cite{Chandra81:alternation}.
This naturally yields a hierarchy of classes delineated by the number of alternations permitted in an accepting run, known as the \emph{polynomial hierarchy} ($\pth$) \cite{Stockmeyer76}, of which both $\nptime$ and $\conptime$ are special cases. 
An almost exact instantiation of this (in a non-uniform setting) is the QBF hierarchy, where formulae are distinguished by their number of quantifier alternations in prefix notation.
This raises the following open-ended question:

\begin{question}
	How do we identify natural fragments of $\pspace$-complete logics complete for levels of the polynomial hierarchy? In particular, can proof theoretic methods help?
\end{question}

In previous work, \cite{Das:17:focussing}, we considered this question for intuitionistic propositional logic, obtaining partial answers for certain expressive fragments.	
In this work we consider the case of \emph{multiplicative additive linear logic} ($\MALL$) \cite{Girard87}, and its affine variant which admits weakening ($\aMALL$); both of these are often seen as the prototypical systems for $\pspace$ since their inference rules constitute the abstract templates of terminating proof search.
Indeed, both $\MALL$ and $\aMALL$ are well-known to be $\pspace$-complete \cite{LincolnMSS90,LincolnMSS92}, results that are subsumed by this work.\footnote{The result for $\aMALL$ is not explicitly stated in \cite{LincolnMSS90,LincolnMSS92}, but it is a folklore result that follows from those methods.}
By considering a `focussed' presentation of $\MALL(\mathsf{w})$, we analyse proof search to identify classes of theorems belonging to each level of $\pth$. 
To demonstrate the accuracy of this method, we also show that these classes are, in fact, \emph{complete} for their respective levels, via encodings from true quantified Boolean formulas (QBFs) of appropriate quantifier complexity, cf.~\cite{Chandra81:alternation}.

The notion of \emph{focussing} is a relatively recent development in structural proof theory that has emerged over the last 20-30 years, \eg\ \cite{Andreoli92,laurent:tel-00007884,liang2009focusing}.
Focussed systems elegantly delineate the phases of invertible and non-invertible inferences in proofs, allowing the natural obtention of {alternating time} bounds for a logic.
Furthermore, they significantly constrain the number of local choices available, resulting in reduced nondeterminism during proof search, while remaining complete.
This result is known as the `focussing' or `focalisation' theorem.
Such systems thus serve as a natural starting point for 
identifying fragments of $\pspace$-complete logics complete for levels of $\pth$.


One shortfall of focussed systems is that, in their usual form, they do not make adequate consideration for \emph{deterministic} computations, which correspond to invertible rules that do not branch, and so the natural measure of complexity there (`decide depth') can considerably overestimate the alternating time complexity of a theorem. 
In the worst case this can lead to rather degenerate bounds, exemplified in \cite{Das:17:focussing} where an encoding of $\sat$ in intuitionistic logic requires a linear decide depth, despite being $\nptime$-complete.\footnote{Indeed, a similar gross overestimation presents in this work if we had used decide depth as our measure of complexity, cf.~Fig.~\ref{fig:ql-der-exists}.}
To deal with this issue \cite{Das:17:focussing} proposed a more controlled form of focussing called \emph{over-focussing}, which allows deterministic steps within synchronous phases, but as noted there this method is not available in $\MALL$ due to the context-splitting $\vlte$ rule.
Instead, in this work we retain the classical abstract notion of focussing, but split the usual invertible, or `asynchronous', phase into a `deterministic' phase, with non-branching invertible rules, and a `co-nondeterministic' phase, with branching invertible rules.
In this way, when expressing proof search as an alternating predicate, a $\forall$ quantifier needs only be introduced in a co-nondeterministic phase.
It turns out that this adaptation suffices to obtain the tight bounds we are after.

This is an extended version of the conference paper \emph{Focussing, $\MALL$ and the polynomial hierarchy} \cite{Das18:ijcar} presented at \emph{IJCAR '18}.
The main differences in this work are the following:
\begin{itemize}
	\item More proof details are provided throughout, in particular for the various intermediate results of Sects.~\ref{sect:qltrans}, \ref{sect:proof-search-predicates} and \ref{sect:lqtrans}.
	\item A whole new section, Sect.~\ref{sect:mallph}, is included which extends the main results of \cite{Das18:ijcar} to pure $\MALL$, i.e.\ \emph{without} weakening.
	\item The exposition is generally expanded, with further commentary and insights throughout.
\end{itemize}
In general, Sects.~2-6 of \cite{Das18:ijcar} cover the same content as their respective sections in this work, although theorem numbers are different.
\smallskip 

This paper is structured as follows. 
In Sect.~\ref{sect:prelims-logic-complexity} we present preliminaries on QBFs and alternating time complexity, and in Sect.~\ref{sect:prelims-ll-proof-search} we present preliminaries on $\MALL(\mathsf w)$ and focussing.
In Sect.~\ref{sect:qltrans} we present an encoding of true QBFs into $\aMALL$, tracking the association between quantifier complexity and alternation complexity of focussed proof search.
In Sect.~\ref{sect:proof-search-predicates} we explain how provability predicates for focussed systems may be obtained as QBFs, with quantifier complexity calibrated appropriately with alternation complexity (the `focussing hierarchy').
In Sect.~\ref{sect:lqtrans} we show how this measure of complexity can be feasibly approximated to yield a bona fide encoding of $\aMALL$ back into true QBFs. Furthermore, we show that the composition of the two encodings preserves quantifier complexity, thus yielding fragments of $\aMALL$ complete for each level of the polynomial hierarchy.
Sect.~\ref{sect:mallph} extends this approach to pure $\MALL$ via carefully composing with a certain encoding of $\aMALL$ into $\MALL$.
Finally, 
in Sect.~\ref{sect:concs} we give some concluding remarks and further perspectives on our presentation of focussing.

%
%
%

%
%
%
%
%

\section{Preliminaries on logic and computational complexity}
\label{sect:prelims-logic-complexity}
In this section we will recall some basic theory of Boolean logic, and its connections to alternating time complexity.
%
%

This section follows Sect.~2 of \cite{Das18:ijcar}, except that we include constants (or `units') for generality here, and we also include a presentation of `Boolean Truth Trees' in Sect.~\ref{sect:sect:BTTs}.

\subsection{Second-order Boolean logic}
\emph{Quantified Boolean formulas} (QBFs) are obtained from the language of classical propositional logic by adding `second-order' quantifiers, varying over propositions.
Formally,
let us fix some set $\Var$ of propositional \emph{variables}, written $x,y$ etc.
QBFs, written $\phi, \psi$ etc., are generated as follows:
\[
\phi
\quad ::= \quad
\cfal \ | \ \ctru \ | \ 
x \  |\ \cnot x \ | \ (\phi \cor \phi) \ | \  (\phi \cand \phi) \ | \ \exists x . \phi \ | \ \forall x . \phi
\]
We write $\cfal$ and $\ctru$ for the classical truth constants \emph{false} and \emph{true} respectively, so that they are not confused with the units from linear logic later.

The formula $\cnot x$ stands for the \emph{negation} of $x$, and all formulas we deal with will be in \emph{De Morgan normal form}, i.e.\ with negation restricted to variables as in the grammar above.
Nonetheless, we may sometimes write $\cnot \phi$ to denote the De Morgan \emph{dual} of $\phi$, generated by the following identities:
\[
\cnot {\cnot x }\ \dfn \ x
\qquad
\begin{array}{rcl}
\cnot \cfal \ &\  \dfn\  &\  \ctru\\
\noalign{\smallskip}
\cnot \ctru \  &\  \dfn \ &\  \cfal
\end{array}
\qquad
\begin{array}{rcl}
\cnot{(\phi \cor \psi)} \ &\ \dfn\  &\ \cnot\phi \cand \cnot\psi
\\
\noalign{\smallskip}
\cnot {(\phi \cand \psi)} \ &\ \dfn\  &\ \cnot\phi \cor \cnot\psi 
\end{array} 
\qquad
\begin{array}{rcl}
\cnot{\exists x . \phi} \ &\ \dfn \ &\ \forall x . \cnot\phi
\\
\noalign{\smallskip}
\cnot{\forall x . \phi} \ &\ \dfn \ &\ \exists x . \cnot\phi
\end{array}
\]
\noindent
A formula is \emph{closed} (or a \emph{sentence}) if all its variables are bound by a quantifier ($\exists$ or $\forall$).
We write $|\phi|$ for the number of occurrences of literals (i.e.\ $x$ or $\cnot x$) in $\phi$.

\medskip

An \emph{assignment} is a function $\alpha : \Var \to \{0,1\}$, here construed as a set $\alpha \subseteq \Var$ in the usual way. We define the \emph{satisfaction} relation between an assignment $\alpha$ and a formula $\phi$, written $\alpha \models \phi$, in the usual way:

\medskip

\noindent\begin{minipage}[t]{.4\textwidth}
	\begin{itemize}
		\item $\alpha \notmodels \cfal$.
		\item $\alpha \models \ctru$.
		
		\smallskip
		
		\item $\alpha \models x$ if $x \in \alpha$.
		\item $\alpha \models \cnot x $ if $x \notin \alpha$.
	\end{itemize}
\end{minipage}
\noindent\begin{minipage}[t]{.55\textwidth}
	\begin{itemize}
		\item $\alpha \models \phi \cor \psi $ if $\alpha \models \phi$ or $\alpha \models \psi$.
		\item $\alpha \models \phi \cand \psi $ if $\alpha \models \phi$ and $\alpha \models \psi$.
		\smallskip
		\item $\alpha \models \exists x . \phi$ if $\alpha \setminus \{x\} \models \phi$ or $\alpha \cup \{x \}\models \phi $.
		\item $\alpha \models \forall  x . \phi$ if $\alpha \setminus \{x\} \models \phi$ and $\alpha \cup \{x \}\models \phi $.
	\end{itemize}
\end{minipage}


\begin{definition}
	[Second-order Boolean logic]
	A QBF $\phi$ is \emph{satisfiable} if there is some assignment $\alpha \subseteq \Var $ such that $\alpha \models \phi$.
	It is \emph{valid} if $\alpha \models \phi$ for every assignment $\alpha \subseteq \Var$.
	If $\phi$ is closed, then we may simply say that it is \emph{true}, written $\models \phi$, when it is satisfiable and/or valid.
	
	\emph{Second-order Boolean logic} ($\qcpl$) is the set of true QBFs.
\end{definition}

In practice, when dealing with a given formula $\phi$, we will only need to consider assignments $\alpha$ that contain variables occurring in $\phi$.
We will assume this later when we discuss predicates (or `languages') computed by open QBFs.

We point out that, from the logical point of view, it suffices to work with only closed QBFs, with satisfiability recovered by prenexing $\exists$ quantifiers and validity recovered by prenexing $\forall$ quantifiers.

\begin{definition}
	[QBF hierarchy]
	For $k \geq 0$ we define the following classes:
	\begin{itemize}
		\item $\Sigma^q_0 = \Pi^q_0$ is the set of quantifier-free QBFs.
		\item $\Sigma^q_{k+1} \supseteq \Pi^q_k$ and, if $\phi \in \Sigma^q_{k+1}$, then so is $\exists x. \phi$. 
		\item $\Pi^q_{k+1} \supseteq \Sigma^q_k$ and, if $\phi \in \Pi^q_{k+1}$, then so is $\forall x. \phi$. 
	\end{itemize}
\end{definition}
Notice that $\phi \in \Sigma^q_k$ if and only if $\cnot \phi \in \Pi^q_k$, by the definition of De Morgan duality.

\smallskip

We have only defined the classes above for `prenexed' QBFs, i.e.\ with all quantifiers at the front. It is well known that any QBF is equivalent to such a formula.
For this reason we will henceforth assume that any QBF we deal with is in prenex normal form. 
In this case we call its quantifier-free part, i.e.\ its largest quantifier-free subformula, the \emph{matrix}.

\subsection{Boolean Truth Trees}
\label{sect:sect:BTTs}

In this work we will not need to formally deal with any deduction system for $\qcpl$, although we point out that there is a simple system whose proof search dynamics closely match quantifier complexity, e.g.\ studied in \cite{Letz02}.
We will briefly present a simplified system in order to exemplify the connection with alternating time complexity.

\emph{Boolean Truth Trees} (BTTs) are a proof system whose lines are closed prenexed QBFs.
Its inference rules are as follows,
\[
\vlinf{\taut}{}{\tau}{}
\qquad
\vlinf{\exists}{}{\exists x . \phi}{\phi[\cfal/x]}
\qquad
\vlinf{\exists}{}{\exists x . \phi}{\phi[\ctru/x]}
\qquad
\vliinf{\forall}{}{\forall x . \phi}{\phi[\cfal/x]}{\phi[\ctru/x]}
\]
where $\tau $ varies over true quantifier-free sentences, i.e.\ true $(\cor,\cand)$-combinations of $\cfal $ and $\ctru$.
Note that we could have further broken down the $ \taut$ rule into several local computation rules, but that is independent of the current analysis.

\begin{example}
	Temporarily write $\vlor$ for the exclusive-or function, i.e.\ $x\vlor y $ is true if either $x$ is true or $y$ is true but not both. 
	The following is a BTT proving $\forall x . \exists y . (x \vlor y)$:
	\[
	\vlderivation{
		\vliin{\forall}{}{\forall x . \exists y . (x \vlor y)}{
			\vlin{\exists}{}{\exists y . (\cfal \vlor y)}{
				\vlin{\taut}{}{\cfal \vlor \ctru}{\vlhy{}}
			}
		}{
			\vlin{\exists}{}{\exists y . (\ctru \vlor y)}{
				\vlin{\taut}{}{\ctru \vlor \cfal}{\vlhy{}}
			}
		}
	}
	\]
\end{example}

Notice that a $\forall$ step is invertible, i.e.\ its conclusion is true just if \emph{every} premiss is true.
On the other hand, an existential formula is true just if \emph{some} $\exists$ step applies.
In this way we can describe the proof search process itself by some `alternating' predicate whose matrix is just a truth-checker for quantifier-free sentences, a deterministic computation.
It is not hard to see that the alternations between $\forall$ and $\exists$ in such a predicate will, in this case, match the quantifier complexity of the input formula, by inspection of the rules.
In order to make all of this more precise, we will need to speak more formally about alternating predicates and alternating complexity.

%

\subsection{Alternating time complexity}
In computation we are used to the distinction between \emph{deterministic} and \emph{nondeterministic} computation.
Intuitively, \emph{co-nondeterminism} is just the `dual' of nondeterminism: at the machine level it is captured by `nondeterministic' Turing machines where \emph{every} run is accepting, not just \emph{some} run as in the case of usual nondeterminism.  
From here \emph{alternating} Turing machines generalise both the nondeterministic and co-nondeterministic models by allowing both universally branching states and existentially branching states.

Intuitions aside, we will now introduce the necessary concepts assuming only a familiarity with deterministic and nondeterministic Turing machines and their complexity measures.
The reader may find a comprehensive introduction to relevant machine models and complexity classes in \cite{papadimitriou94}.

\smallskip

For a language $L$ of strings over some finite alphabet, we write $\nptime (L)$ for the class of languages accepted in polynomial time by some nondeterministic Turing machine which may, at any point, query in constant time whether some word is in $L$ or not.
We extend this to classes of languages $\mathcal C$, writing $\nptime (\mathcal C)$ for $\bigcup\limits_{L \in \mathcal C} \nptime (L)$.
We also write $\mathbf{co}\mathcal C$ for the class of languages whose complements are in $\mathcal C$.
\begin{definition}
	[Polynomial hierarchy, \cite{Stockmeyer76}]
	\label{dfn:ph}
	We define the following classes:
	\begin{itemize}
		\item $\Sigma^p_0 = \Pi^p_0 \dfn \ptime$.
		\item $\Sigma^p_{k+1} \dfn \nptime (\Sigma^p_k)$.
		\item $\Pi^p_{k+1} \dfn \mathbf{co}\Sigma^p_{k+1}$.
	\end{itemize}
	The \emph{polynomial hierarchy ($\pth$)} is $\bigcup\limits_{k = 0}^\infty \Sigma^p_k = \bigcup\limits_{k = 0}^\infty \Pi^p_k$.
\end{definition}


We may more naturally view the polynomial hierarchy as the bounded-quantifier-alternation fragments of QBFs we introduced earlier.
%
%
%
For this we construe $\Sigma^q_k$ and $\Pi^q_k$ as classes of finite \emph{languages}, by associating with a QBF $\phi(x_1, \dots, x_n)$ (with all free variables indicated) the class of (finite) assignments $\alpha \subseteq \{x_1, \dots , x_n \}$ satisfying it. These assignments may themselves may be seen as binary strings of length $n$ which encode their characteristic functions in the usual way.

\begin{definition}
	[Evaluation problems]
	Let $\mathcal C $ be a set of QBFs.
	\emph{$\mathcal C$-evaluation} is the problem of deciding, given a formula $\phi (\vec x) \in \mathcal C$, with all free variables indicated, and an assignment $\alpha \subseteq \vec x$, whether $\alpha \models \phi (\vec x)$.
\end{definition}

\begin{theorem}
	[cf.~\cite{Chandra81:alternation}]
	\label{thm:qbf-complexity}
	For $k \geq 1$, we have the following:
	\begin{enumerate}
		\item $\Sigma^q_k $-evaluation is $\Sigma^p_k$-complete.
		\item $\Pi^q_k$-evaluation is $\Pi^p_k$-complete.
	\end{enumerate}
\end{theorem}


\begin{corollary}
	\label{cor:qbf-complete}
	For $k\geq 1$, we have the following:
	\begin{enumerate}
		\item $\{ \phi \in \Sigma^q_k : \text{$\phi$ is closed and true} \}$ is $\Sigma^p_k$-complete.
		\item $\{ \phi \in \Pi^q_k : \text{$\phi$ is closed and true} \}$ is $\Pi^p_k$-complete.
	\end{enumerate}
	
\end{corollary}

\begin{proof}
	Membership is immediate from Thm.~\ref{thm:qbf-complexity}, evaluating under the assignment $\emptyset$. For hardness, notice that we may always simplify a QBF under an assignment $\alpha$ to a closed formula as follows: first, replace all free variable occurrences $x$ with $\ctru$ if $x \in \alpha$ and $\cfal $ otherwise. Now simply apply the following rewrite rules,
	\begin{equation}
	\label{eqn:simplifying-qbf}
	\begin{array}{rccclcrcccl}
	\cfal \cor \phi &\rightarrow &\phi &\leftarrow & \phi \cor \cfal &\qquad &\cfal \cand \phi &\rightarrow &\cfal &\leftarrow &\phi \cand \cfal\\
	\noalign{\smallskip}
	\ctru \cor \phi &\rightarrow &\ctru &\leftarrow &\phi \cor \ctru &\qquad & \ctru \cand \phi &\rightarrow &\phi &\leftarrow &\phi \cand \ctru
	\end{array}
	\qedhere
	\end{equation}
\end{proof}

\section{Linear logic and proof search}
\label{sect:prelims-ll-proof-search}

\emph{Linear logic} was introduced by Girard \cite{Girard87} to decompose the mechanics of cut-elimination by means of different connectives.
It naturally subsumes both classical and intuitionistic logic by various embeddings, and has furthermore been influential in the theoretical foundations of \emph{logic programming} via the study of \emph{focussing}, which constrains the level of nondeterminism in proof search, cf.~\cite{Andreoli92,DelandeMS10,ChaudhuriMS08}.
In this work we only consider the fragment \emph{multiplicative additive linear logic} ($\MALL$) and its version with `weakening'($\aMALL$).

This section mostly follows Sect.~3 of \cite{Das18:ijcar}, 
mainly differing in that we here include units in the formulation of $\MALL$ and $\aMALL$, for generality, and give some further proof details.

\subsection{Multiplicative additive linear logic}


For convenience, we work with the same set $\Var$ of variables that we used for QBFs.
To distinguish them from QBFs, we use the metavariables $A,B, $ etc.\ for $\MALL(\mathsf w)$ formulas, generated as follows:
\[
A \quad ::= \quad
\bot \ | \ 0 \ | \ 1 \ | \ \top\ |\
x \ | \ \lneg x \ | \ (A \vlpa B) \ | \ (A \vlor B) \ | \ (A \vlte B) \ | \  (A \vlan B)
\]


\noindent
$
\bot, 1 , 
\vlpa, \vlte$ are called \emph{multiplicative} connectives, and $
0, \top, 
\vlor, \vlan$ are called \emph{additive} connectives.
%
Like for QBFs, we have restricted negation to the variables, thanks to De Morgan duality in $\MALL$.
Again, we may write $\lneg A$ for the De Morgan dual of $A$, which is generated similarly to the case of QBFs:
\[
\cnot {\cnot x }\ \dfn \ x
\quad
\begin{array}{rcl}
\cnot \bot  &  \dfn  &  1\\
\noalign{\smallskip}
\cnot 1   &  \dfn &  \bot\\
\noalign{\medskip}
\cnot 0  &  \dfn  &  \top \\
\noalign{\smallskip}
\cnot \top   &  \dfn &  0
\end{array}
\quad
\begin{array}{rcl}
\cnot{(A \vlpa B)}  & \dfn & \cnot A \vlte \cnot B
\\
\noalign{\smallskip}
\cnot {(A \vlte B)}  & \dfn  & \cnot A \vlpa \cnot B
\end{array} 
\quad
\begin{array}{rcl}
\cnot{(A \vlor B)}  & \dfn  & \cnot A \vlan \cnot B
\\
\noalign{\smallskip}
\cnot {(A \vlan B)}  & \dfn  & \cnot A \vlor \cnot B 
\end{array}
\]
%


Due to De Morgan duality, we will work only with `one-sided' calculi for $\MALL$ and $\aMALL$, where all formulas occur to the right of the sequent arrow.
This means we will have fewer cases to consider for formal proofs, although later we will also informally adopt a two-sided notation when it is convenient, cf.~Rmk.~\ref{not:two-sided}.

\begin{definition}
	[$\MALL(\mathsf w)$]
	A \emph{cedent}, written $\Gamma, \Delta$ etc., is a multiset of formulas, delimited by commas `,', and a \emph{sequent} is an expression $\seqar \Gamma$.\footnote{We will often identify cedents and sequents, since we are in a one-sided setting.} 
	The system (cut-free) $\MALL$ is given in Fig.~\ref{fig:mall}.
	$\aMALL$, a.k.a.\ \emph{affine} $\MALL$, is defined in the same way, only with the $(\id)$ rule and $(1)$ rule replaced by the following analogues:
	\begin{equation}
	\label{eqn:wkid}
	\vlinf{\wkid}{}{\seqar \Gamma, x , \lneg x}{}
	\qquad
	\vlinf{w1}{}{\seqar \Gamma, 1}{}
	\end{equation}
\end{definition}
\begin{figure}[t]
	\[
	\vlinf{\id}{}{\seqar x , \lneg x}{}
	\qquad
	\begin{array}{cccc}
	\vlinf{\bot}{}{\seqar \Gamma, \bot}{\seqar \Gamma} 
	\quad &\quad \vlinf{1}{}{\seqar  1}{}
	\quad &\quad  \vlinf{\vlpa}{}{\seqar \Gamma, A \vlpa B}{\seqar \Gamma, A , B}
	\quad &\quad  \vliinf{\vlte}{}{\seqar \Gamma, \Delta, A \vlte B}{\seqar \Gamma, A}{\seqar \Delta, B} 
	\\
	\noalign{\smallskip}
	{\color{gray}
		\ \ \left(
		\begin{array}{c}
		\text{no rule}\\
		\text{for $0$}
		\end{array}
		\right)
	}
	\quad &\quad  \vlinf{\top}{}{\seqar \Gamma, \top}{}
	\quad &\quad  \vlinf{\vlor}{}{\seqar \Gamma, A_0 \vlor A_1}{\seqar \Gamma, A_i} 
	\quad &\quad  \vliinf{\vlan}{}{\seqar \Gamma, A \vlan B}{\seqar \Gamma , A}{\seqar \Gamma , B} 
	\end{array}
	\]
	\caption{The system (cut-free) $\MALL$, where $i \in \{0,1 \}$.}
	\label{fig:mall}
\end{figure}

\noindent
We have not included the `cut' rule, thanks to cut-elimination for linear logic \cite{Girard87}.
We will only study cut-free proofs in this paper.
Notice that, following the tradition in linear logic, we write `$\seqar$' for the sequent arrow, though we point out that the deduction theorem does not actually hold w.r.t.\ linear implication.
For the affine variant, we have simply built weakening into the initial steps, since it may always be permuted upwards in a proof:
\begin{proposition}
	[Weakening admissibility]
	\label{prop:wk-admiss-mall}
	The following rule, called 
	\emph{weakening}, is admissible in $\aMALL$:
	\[
	\vlinf{\wk}{}{\seqar \Gamma, A}{\seqar \Gamma}
	\]
\end{proposition}
\begin{proof}
	This is a routine (and indeed well-known) argument by induction on the size of a subproof that roots a weakening step.
	The initial sequents of $\aMALL$ are already closed under weakening, and have the following inductive cases:
	\[
	\begin{array}{ccc}
	\noalign{\medskip}
	\vlderivation{
		\vlin{\wk}{}{\seqar \Gamma , A \vlpa B , C }{
			\vlin{\vlpa}{}{\seqar \Gamma , A\vlpa B}{\vlhy{\seqar \Gamma , A , B}}
		}
	}
	\quad &\quad  \leadsto \quad&\quad 
	\vlderivation{
		\vlin{\vlpa}{}{\seqar \Gamma , A \vlpa B , C }{
			\vlin{\wk}{}{\seqar \Gamma , A,  B, C}{\vlhy{\seqar \Gamma , A , B}}
		}
	}
	\\
	\noalign{\medskip}
	\vlderivation{
		\vlin{\wk}{}{\seqar \Gamma , A_0 \vlor A_1, C}{
			\vlin{\vlor}{}{\seqar \Gamma , A_0 \vlor A_1}{\vlhy{\seqar \Gamma , A_i}}
		}
	}
	\quad &\quad  \leadsto \quad&\quad 
	\vlderivation{
		\vlin{\vlor}{}{\seqar \Gamma , A_0 \vlor A_1, C}{
			\vlin{\wk}{}{\seqar \Gamma , A_0 , A_1,C}{\vlhy{\seqar \Gamma , A_i}}
		}
	}
	\\
	\noalign{\medskip}
	\vlderivation{
		\vlin{\wk}{}{\seqar \Gamma , \Delta, A\vlte B , C}{
			\vliin{\vlte}{}{\seqar \Gamma , \Delta, A\vlte B}{\vlhy{\seqar \Gamma , A}}{\vlhy{\Delta, B}}
		}
	}
	\quad &\quad  \leadsto \quad&\quad 
	\vlderivation{
		\vliin{\vlte}{}{\seqar \Gamma , \Delta, A\vlte B , C}{
			\vlin{\wk}{}{\seqar \Gamma , A, C}{\vlhy{\seqar \Gamma , A}}
		}{\vlhy{\seqar \Delta , B}}
	}
	\\
	\noalign{\medskip}
	\vlderivation{
		\vlin{\wk}{}{\seqar \Gamma , A\vlan B , C}{
			\vliin{\vlan}{}{\seqar \Gamma , A \vlan B}{\vlhy{\seqar \Gamma , A}}{\vlhy{\seqar \Gamma , B}}
		}
	}
	\quad &\quad  \leadsto \quad&\quad 
	\vlderivation{
		\vliin{\vlan}{}{\seqar \Gamma , A \vlan B, C}{
			\vlin{\wk}{}{\seqar \Gamma , A, C}{\vlhy{\seqar \Gamma , A}}
		}{
			\vlin{\wk}{}{\seqar \Gamma , B , C}{\vlhy{\seqar \Gamma  , B}}
		}
	}
	\end{array}
	\qedhere
	\]
\end{proof}


\subsection{(Multi-)focussed systems for proof search}
Focussed systems for $\MALL$ (and linear logic in general) have been widely studied \cite{Andreoli92,laurent:tel-00007884,DelandeMS10,ChaudhuriMS08}.
The idea is to associate polarities to the connectives based on whether their introduction rule is invertible (negative) or their dual's introduction rule is invertible (positive).
Now bottom-up proof search can be organised in a manner where, once we have chosen a positive principal formula to decompose (the `focus'), we may continue to decompose its auxiliary formulas until the focus becomes negative. The main result herein is the \emph{completeness} of such proof search strategies, known as the \emph{focussing theorem} (a.k.a.\ the `focalisation theorem').

It is known that `multi-focussed' variants, where one may have many foci in parallel, lead to certain `canonical' representations of proofs for $\MALL$ \cite{ChaudhuriMS08}.
Furthermore, the alternation behaviour of focussed proof search can be understood via a game theoretic approach \cite{DelandeMS10}.
However, such frameworks unfortunately fall short of characterising the alternating complexity of proof search in a faithful way.	
The issue is that the usual focussing methodology does not make any account for \emph{deterministic} computations, which correspond to invertible rules that do not branch.
Such rules are usually treated just like the other invertible rules, which in general comprise the `co-nondeterministic' stages of proof search.


For these reasons we introduce a bespoke presentation of (multi-)focussing for $\MALL$, with a designated \emph{deterministic} phase dedicated to invertible non-branching rules, in particular the $\vlpa$ rule.
To avoid conflicts with more traditional presentations, we call the other two phases \emph{nondeterministic} and \emph{co-nondeterministic} rather than `synchronous' and `asynchronous' respectively. This terminology also reinforces the intended connections to computational complexity.


Henceforth we use $a,b,$ etc.\ to vary over atomic formulas. 
We also use the following metavariables to vary over formulas with the indicated main connectives:
\[
\begin{array}{rcll}
M & : & \text{`negative and not deterministic'} &\quad \vlan \\
N & : & \text{`negative'} & \quad \vlan, \vlpa \\
O & : & \text{`deterministic'} &\quad \vlte , \vlpa , a
\\
P & : & \text{`positive'} & \quad \vlte, \vlor \\
Q & : & \text{`positive and not deterministic'} &  \quad \phantom{ \vlte , } \hspace{0.2em} \vlor
\end{array}
\]
`Vectors' are used to vary over multisets of associated formulas, e.g.\ $\vec P$ varies over multisets of $P$-formulas. We may sometimes view these as sequences or even sets for convenience.
Sequents may now contain a single delimiter $\Downarrow$ or $\Uparrow$.

\begin{definition}
	[Multi-focussed proof system]
	\label{dfn:multifoc-mallw}
	We define the (multi-focussed) system $\focMALL$ in Fig.~\ref{fig:focmall}.
	The system $\afocMALL$ is the same as $\focMALL$ but with the $(\id) $ and $(1)$ rules replaced by the rules $(\wkid)$ and $(w1)$ from \eqref{eqn:wkid}.
\end{definition}

\begin{figure}[t]
	\noindent
	\textbf{Deterministic phase:}
	\[
	\vlinf{\id}{}{\seqar x,\lneg{x}}{}
	\quad
	\vlinf{\bot}{}{\seqar \Gamma, \bot}{\seqar \Gamma}
	\quad
	\vlinf{1}{}{\seqar  1}{}
	\quad
	\vlinf{\top}{}{\Gamma, \top}{}
	\quad
	\vlinf{\vlpa}{}{\seqar \Gamma, A \vlpa B}{\seqar \Gamma, A , B}
	\quad
	\vlinf{\dec}{
	}{\seqar \vec a, \vec{P}, \vec P'}{\seqar \vec a , \vec{P}  \Downarrow \vec P'}
	\quad
	\vlinf{\codec}{
	}{\seqar \vec a , \vec{P}, \vec M}{\seqar \vec a , \vec{P} \Uparrow \vec M}
	\]

	\medskip
	
	\noindent
	\textbf{Nondeterministic phase:}
	\[
	\vlinf{\vlor}{}{\seqar \Gamma \Downarrow\Delta,  A_0 \vlor A_1}{\seqar \Gamma \Downarrow\Delta,  A_i}
	\qquad
	\vliinf{\vlte}{}{\seqar \Gamma, \Delta \Downarrow \Sigma, \Pi, A \vlte B}{\seqar \Gamma \Downarrow \Sigma,  A}{\seqar \Delta \Downarrow \Pi,  B}
	\qquad
	\vlinf{\rel}{}{\seqar \Gamma \Downarrow \vec a , \vec{N} }{\seqar \Gamma, \vec a , \vec{N} }
	\]

	\medskip
	
	\noindent
	\textbf{Co-nondeterministic phase:}
	\[
	\vliinf{\vlan}{}{\seqar \Gamma \Uparrow \Delta, A \vlan B}{\seqar \Gamma \Uparrow \Delta ,A}{\Gamma \Uparrow \Delta , B}
	\qquad
	\vlinf{\corel}{}{\seqar \Gamma \Uparrow  \vec P, \vec O }{\seqar \Gamma,  \vec P , \vec O }
	\]
	\caption{The system (cut-free) $\focMALL$, where $\vec P'$ and $\vec M$ must be nonempty and $i \in \{0,1\}$.}
	\label{fig:focmall}
\end{figure}

A proof of a formula $A$ is simply a proof of the sequent $\seqar A$, i.e.\ there is no need to pre-decorate with arrows, as opposed to usual presentations, thanks to the deterministic phase.
The rules $\dec$ and $\codec$ are called \emph{decide} and \emph{co-decide} respectively, while
$\rel$ and $\corel$ are called \emph{release} and \emph{co-release} respectively. 
%

Notice that the determinism of $\vlte$ plays no role in this one-sided calculus, but in a two-sided calculus we would have a deterministic left $\vlte$ rule that is analogous to the given $\vlpa$ rule (on the right).
This is the same as how the `negativity' (in the sense of non-invertibility on the left) of $\vlpa$ plays no role in this calculus.
One may argue that a $\vlte$ on the left is morally just a $\vlpa$, but such a simplification sacrifices the duality of connectives being reflected in terms of duality in computational complexity: if $\vlpa$ is deterministic, then so should be its dual, $\vlte$.
Indeed, in the above classification, $O$-formulas are dual to $O$-formulas, $P$-formulas dual to $N$-formulas, and $Q$-formulas dual to $M$-formulas.


\smallskip

As usual for multi-focussed systems, the analogous focussed system can be recovered by restricting rules to only one focussed formula in nondeterministic phases.
Moreover, in our presentation, we may also impose the \emph{dual} restriction, that there is only one formula in `co-focus' during a co-nondeterministic phase:

\begin{definition}
	[Simply (co-)focussed subsystems]
	\label{dfn:co-bi-focussed}
	A $\focMALL$ proof is \emph{focussed} if $\vec P'$ in $\dec$ is always a singleton.
	It is \emph{co-focussed} if $\vec M$ in $\codec$ is always a singleton.
	If a proof is both focussed and co-focussed then we say it is \emph{bi-focussed}.
\end{definition}

\noindent
The notion of `co-focussing' is not usually possible for (multi-)focussed systems since the invariant of being a singleton is not usually maintained in an asynchronous phase, due to the $\vlpa$ rule. 
However we treat $\vlpa$ as deterministic rather than co-nondeterministic, and we can see that the $\vlan$-rule indeed maintains the invariant of having just one formula on the right of $\Uparrow$.

\begin{theorem}
	[Focussing theorem]
	\label{thm:bi-foc-complete}
	We have the following:
	\begin{enumerate}
		\item The class of bi-focussed $\focMALL$-proofs is complete for $\MALL$.
		\item The class of bi-focussed $\afocMALL$-proofs is complete for $\aMALL$.
	\end{enumerate}
\end{theorem}
Evidently, this immediately means that $\focMALL$ ($\afocMALL$), as well as its focussed and co-focussed subsystems, are also complete for $\MALL$ (resp.~$\aMALL$).
The proof of Thm.~\ref{thm:bi-foc-complete} follows routinely from any other completeness proof for focussed $\MALL$, e.g.~\cite{Andreoli92,laurent:tel-00007884}. 
The only change in our presentation is in the organisation of phases, for which we may think of bi-focussed proofs as certain annotated focussed proofs.

\medskip

To aid our exposition, we will sometimes use a `two-sided' notation and extra connectives so that the intended semantics of sequents are clearer.
Strictly speaking, this is just a shorthand for one-sided sequents: the calculi defined in Figs.~\ref{fig:mall} and \ref{fig:focmall} are the formal systems we are studying.

\begin{notation}[Two-sided notation]
	\label{not:two-sided}
	We write $\Gamma \seqar \Delta $ as shorthand for the sequent $\seqar \lneg{\Gamma}, \Delta$, where $\lneg \Gamma $ is $\{ \lneg A : A \in \Gamma \}$.
	We extend this notation to sequents with $\Uparrow$ or $\Downarrow$ symbols in the natural way, writing $\Gamma \Uparrow \Delta \seqar \Sigma \Uparrow \Pi$ for $\seqar \lneg \Gamma, \Sigma \Uparrow \lneg \Delta, \Pi$ and $\Gamma\Downarrow \Delta  \seqar \Sigma \Downarrow \Pi$ for $\seqar \lneg \Gamma, \Sigma \Downarrow \lneg \Delta, \Pi $.
	In all cases, (co-)foci are always written to the right of $\Downarrow$ or $\Uparrow$.
	
	\smallskip
	
	We write $A \limp B $ as shorthand for the formula $\lneg A \vlpa B$, and $A\plimp B$ as shorthand for the formula $\lneg A \vlor B$.
	Sometimes we will write, e.g., a step,
	\[
	\vliinf{\lefrul \limp}{}{\Gamma, \Gamma'\Downarrow A \limp B \seqar \Delta, \Delta'}{\Gamma \seqar \Delta\Downarrow A  }{\Gamma' \Downarrow B \seqar \Delta'}
	\]
	which, by
	definition, 
	corresponds to a correct application of ${\vlte}$ in $\focMALL(\mathsf w)$.
\end{notation}

\section{An encoding from $\qcpl$ to $\aMALL$}
\label{sect:qltrans}

%
%
%

In this section we present an encoding of true QBFs into $\aMALL$.
(We will later adapt this into an encoding into $\MALL$ in Sect.~\ref{sect:mallph}.)
The former were also used for the original proof that $\MALL$ is $\pspace$-complete \cite{LincolnMSS90,LincolnMSS92}, though our encoding differs from theirs and leads to a more refined result, cf.~Sect.~\ref{sect:lqtrans}. 

This section mostly follows Sect.~4 from \cite{Das18:ijcar}, except with some further details in proofs and the exposition.
Henceforth we assume that all QBFs are in prenex normal form and free of truth constants $\cfal$ and $ \ctru$ (e.g.\ by Eqn.~\ref{eqn:simplifying-qbf}).

\subsection{Positive and negative encodings of quantifier-free evaluation}

The base cases of our translation from QBFs to $\aMALL$ will be quantifier-free Boolean formula evaluation.
This is naturally a deterministic computation,
being polynomial-time computable.
(In fact, Boolean formula evaluation is known to be $\mathbf{ALOGTIME}$-complete \cite{Buss87:formula-eval-nc1}.)
However one issue is that this determinism cannot be seen from the point of view of $\aMALL$, since the only deterministic connective ($\vlpa$, on the right) is not expressive enough to encode evaluation.

Nonetheless we are able to circumvent this problem since $\aMALL$ is at least able to see that quantifier-free evaluation is in $\nptime \cap \conptime$, via a pair of corresponding encodings.
For non-base levels of $\pth$ this is morally the same as being deterministic, as we will see more formally over the course of this section.

\begin{definition}
	[Positive and negative encodings]
	Let $\phi$ be a quantifier-free Boolean formula. We define:
	\begin{itemize}
		\item $\negtrans{\phi}$ is the result of replacing every $\cor $ in $\phi$ by $\vlpa$ and every $\cand $ in $\phi$ by $\vlan$.
		\item $\postrans{\phi}$ is the result of replacing every $\cor $ in $\phi$ by $\vlor$ and every $\cand $ in $\phi$ by $\vlte$.
	\end{itemize}
	
	For an assignment $\alpha\subseteq \Var$ and a list of variables $\vec x = (x_1 , \dots , x_k)$, we write $\alpha(\vec x)$ for the cedent $\{ x_i :  x_i \in \alpha ,  i\leq k \} \cup \{ \lneg{x}_i : x_i \notin \alpha , i\leq k \}$.
	We write $\alpha^n(\vec x)$ for the cedent consisting of $n$ copies of each literal in $\alpha(\vec x)$.
\end{definition}

\begin{proposition}
	\label{prop:pos-neg-sat}
	Let $\phi$ be a quantifier-free Boolean formula with free variables $\vec x$ and let $\alpha $ be an assignment.
	For $n \geq |\phi|$, the following are equivalent:
	\begin{enumerate}
		\item\label{item:satisfaction} $\alpha \models \phi$.
		\item\label{item:negsat} $\aMALL$ proves $\alpha(\vec x) \seqar \negtrans \phi$.
		\item\label{item:possat} $\aMALL$ proves $\alpha^n(\vec x) \seqar \postrans \phi$.
	\end{enumerate}
\end{proposition}


\begin{proof}
	\ref{item:negsat}$\implies$\ref{item:satisfaction} and \ref{item:possat}$\implies$\ref{item:satisfaction} are immediate from the `soundness' of $\aMALL$ with respect to classical logic, by interpreting $\vlte $ or $\vlan $ as $\cand$ and $\vlor$ or $\vlpa$ as $\cor$.

	Intuitively, \ref{item:satisfaction}$\implies$\ref{item:negsat} follows directly from the invertibility of rules, while for \ref{item:satisfaction}$\implies$\ref{item:possat} we may appeal to the usual properties of satisfaction while controlling linearity appropriately.
	Formally we prove the following more general statements:
	\begin{itemize}
		\item For any multiset $\Lambda$ of quantifier-free Boolean formulas, if $\alpha \models \bigvee \Lambda$ then $\aMALL$ proves $\alpha (\vec x ) \seqar \negtrans \Lambda $, where $\negtrans \Lambda $ is the $\aMALL$ cedent $\{ \negtrans{\phi} : \phi \in \Lambda \}$.
		\item For any quantifier-free Boolean formula $\phi$ with $|\phi|\leq n$, if $\alpha \models \phi$ then $\aMALL$ proves $\alpha^n(\vec x) \seqar \postrans{\phi}$.
	\end{itemize}
	
	\noindent
	We proceed by induction on the number of connectives in $\Lambda $ or $\phi$.
	The bases case is simple (relying on affinity)
	and the inductive cases are as follows,
	\[
	\begin{array}{cc}
	\vlderivation{
		\vlin{\vlpa}{}{\alpha(\vec x) \seqar \Lambda, \negtrans \phi \vlpa \negtrans \psi }{
			\vltr{\IH}{\alpha (\vec x) \seqar \Lambda, \negtrans \phi, \negtrans \psi}{\vlhy{\ }}{\vlhy{\ }}{\vlhy{\ }}
		} 
	} 
	\quad &\quad 
	\vlderivation{
		\vliin{\vlan}{}{\alpha (\vec x) \seqar \Lambda, \negtrans \phi \vlan \negtrans \psi}{ 
			\vltr{\IH}{\alpha (\vec x) \seqar \Lambda, \negtrans \phi}{\vlhy{\ }}{\vlhy{\ }}{\vlhy{\ }}
		}{
			\vltr{\IH}{\alpha (\vec x) \seqar \Lambda, \negtrans \psi}{\vlhy{\ }}{\vlhy{\ }}{\vlhy{ \ }}
		} 
	}
	\\
	\noalign{\medskip}
	\vlderivation{
		\vlin{\vlor}{}{\alpha^n (\vec x) \seqar \postrans{\phi_0} \vlor \postrans{\phi_1}}{
			\vltr{\IH}{\alpha^n (\vec x) \seqar \postrans \phi_i }{\vlhy{\ }}{\vlhy{\ }}{\vlhy{\ }}
		} 
	}
	\quad &\quad 
	\vlderivation{
		\vliin{\vlte}{}{\alpha^n (\vec x) \seqar \postrans \phi \vlte \postrans \psi}{
			\vltr{\IH}{\alpha^l (\vec x) \seqar \postrans \phi}{\vlhy{\ }}{\vlhy{\ }}{\vlhy{\ }}
		}{
			\vltr{\IH}{\alpha^m (\vec x) \seqar \postrans \psi}{\vlhy{\ }}{\vlhy{\ }}{\vlhy{\ }}
		} 
	}
	\end{array}
	\]
	where: 
	\begin{itemize}
		\item $i = 0$ or $i=1$, depending on whether $\alpha \models \phi_0$ or $\alpha \models \phi_1$, respectively; and,
		\item $l$ and $m$ are chosen so that $l \geq |\phi|$ and $m\geq |\psi|$; and,
		\item the derivations marked $\IH$ are obtained from the inductive hypothesis.\qedhere
	\end{itemize}
	
\end{proof}

\subsection{Encoding quantifiers in $\aMALL$}
As we said before, we do not follow the `locks-and-keys' approach of \cite{LincolnMSS90,LincolnMSS92}. 
Instead we follow a similar approach to Statman's proof that intuitionistic propositional logic is $\pspace$-hard \cite{Statman79}, adapted to minimise proof search complexity.

The basic idea is that we would like to encode quantifiers as follows:
\begin{equation}
\label{eqn:quant-to-prop-idea}
\begin{array}{rcl}
\exists x . \phi \quad  &\leadsto &\quad (x\limp \phi ) \vlor (\cnot x \limp \phi) 
\\
\forall x . \phi \quad & \leadsto & \quad (x \limp \phi) \vlan (\cnot x\limp \phi)
\end{array}
\end{equation}
The issue is that such a naive approach would induce an exponential blowup, due to the two occurrences of $\phi$ in each line above.
This idea was considered by Statman in \cite{Statman79}, for intuitionistic propositional logic, where he avoided the blowup by using \emph{Tseitin extension variables}, essentially fresh variables used to abbreviate complex formulas, e.g.\ $(x \equiv \phi)$.
The issue is that this can considerably complicate the structure of proofs, since, in order to access the abbreviated formula, we must pass both a positive and negative phase induced by $\equiv$.

Instead, we use an observation from \cite{Das:17:focussing} that $\phi$ occurs only positively in \eqref{eqn:quant-to-prop-idea} above, and so we only need one direction of Tseitin extension.
Doing this carefully will allow us to control the structure proofs in a way that is consistent with the alternation complexity of the initial QBF, as we will see later.

\begin{definition}
	[$\qcpl$ to $\aMALL$]
	Given a QBF $\phi = Q_k x_k . \cdots . Q_1 x_1 . \phi_0$ with $|\phi_0|=n$ and all quantifiers indicated, we define $\qltrans{\phi}$ by induction on $k$ as follows,
	\[
	\begin{array}{rcl}
	\qltrans{\phi_0} \ &\  \dfn \ &\  \begin{cases}
	\postrans{\phi_0}\  & \ \text{if $Q_1$ is $\exists$} \\
	\negtrans{ \phi_0}\  &\  \text{if $Q_1$ is $\forall$}
	\end{cases}
	\\
	\noalign{\medskip}
	\qltrans{Q_{k} x_k . \phi'} \ &\  \dfn \ &\  
	\begin{cases}
	(\qltrans{ \phi'} \ \limp\  y_k) \limp ((x_k^n \limp y_k) \vlor (\lneg x_k^n \limp y_k) ) & \text{if $Q_{k}$ is $\exists$} \\
	(\qltrans {\phi'} \plimp y_k) \limp ((x_k^n \limp y_k) \vlan (\lneg x_k^n \limp y_k) ) & \text{if $Q_{k}$ is $\forall$}
	\end{cases}
	\end{array}
	\]
	where $y_k$ is always fresh. 
\end{definition}

\begin{lemma}
	\label{lem:qltrans-correct}
	Let $\phi(\vec x)$ be a QBF with all free variables displayed and matrix $\phi_0$, with $|\phi_0| =n $.
	Then $\alpha\models \phi$ if and only if $\aMALL$ proves $\alpha^n (\vec x) \seqar \vec y, \qltrans \phi $ for any assignment $\alpha$ and any $\vec y$ disjoint from $\vec x$.
\end{lemma}

\begin{figure}[t]
	\[
	\vlderivation{
		\vlin{=}{}{\alpha^n (\vec x) \seqar \vec y, \qltrans{\exists x . \phi}  }{
			\vlin{\rigrul \limp}{}{\alpha^n (\vec x) \seqar \vec y, ( \qltrans \phi \limp y) \limp ((x^n \limp y) \vlor (\lneg x^n \limp y) ) }{
				\vlin{\rigrul{\dec}}{}{\alpha^n (\vec x) , \qltrans \phi \limp y \seqar  \vec y,(x^n \limp y) \vlor (\lneg x^n \limp y)}{
					\vlin{\rigrul{\vlor}}{}{\alpha^n (\vec x) , \qltrans \phi \limp y \seqar \vec y \Downarrow (x^n \limp y) \vlor (\lneg x^n \limp y) }{
						\vlin{\rigrul{\rel}}{}{\alpha^n (\vec x) , \qltrans \phi \limp y \seqar  \vec y \Downarrow \pm x^n \limp y }{
							\vlin{\rigrul{\limp}}{}{\alpha^n (\vec x) , \qltrans \phi \limp y \seqar \vec y , \pm x^n \limp y}{
								\vlin{\lefrul \dec }{}{\alpha^n (\vec x) , \pm x^n, \qltrans \phi \limp y \seqar \vec y , y}{
									\vliin{\lefrul \limp}{}{\alpha^n (\vec x) , \pm x^n \Downarrow \qltrans \phi \limp y \seqar \vec y, y }{
										\vlin{\rigrul \rel}{}{\alpha^n (\vec x) , \pm x^n \seqar \vec y \Downarrow \qltrans \phi  }{
											\vltr{\IH}{\alpha^n (\vec x) , \pm x^n \seqar  \vec y, \qltrans \phi }{\vlhy{\quad }}{\vlhy{}}{\vlhy{\quad }}
										}
									}{
										\vlin{\lefrul \rel}{}{\Downarrow y \seqar y}{
											\vlin{\id}{}{y \seqar y}{\vlhy{}}
										}
									}
								}
							}
						}	
					}
				}
			}
		}	
	}
	\]
	\caption{Proof of $\exists$ case for left-right direction of Lemma~\ref{lem:qltrans-correct}.}
	\label{fig:ql-der-exists}
\end{figure}
\begin{figure}[t]
	\[
	\vlderivation{
		\vlin{=}{}{\alpha^n (\vec x) \seqar \vec y, \qltrans{\forall x . \phi} }{
			\vlin{\rigrul{\limp}}{}{\alpha^n (\vec x) \seqar \vec y, (\qltrans \phi \plimp y ) \limp ((x^n \limp y) \vlan (\lneg x^n \limp y)) }{
				\vlin{\rigrul{\codec}}{}{\alpha^n (\vec x) , \qltrans \phi \plimp y \seqar  \vec y, (x^n \limp y) \vlan (\lneg x^n \limp y) }{
					\vliin{\rigrul{\vlan}}{}{\alpha^n (\vec x) , \qltrans \phi \plimp y \seqar \vec y\Uparrow (x^n \limp y) \vlan (\lneg x^n \limp y) }{
						\vlin{\rigrul{\corel}}{}{\alpha^n (\vec x) , \qltrans \phi \plimp y \seqar \vec y \Uparrow x^n \limp y }{
							\vlin{\rigrul{\limp}}{}{\alpha^n (\vec x) , \qltrans \phi \plimp y \seqar \vec y,  x^n \limp y}{
								\vlin{\lefrul \codec }{}{\alpha^n (\vec x) , x^n, \qltrans \phi \plimp y \seqar  \vec y, y }{
									\vliin{\lefrul \plimp }{}{\alpha^n (\vec x) , x^n \Uparrow\qltrans \phi \plimp y \seqar \vec y, y }{
										\vlin{\rigrul \corel}{}{ \alpha^n (\vec x) , x^n \seqar \vec y, y \Uparrow\qltrans \phi  }{
											\vltr{\IH}{\alpha^n (\vec x), x^n \seqar  \vec y, y, \qltrans\phi }{\vlhy{\quad }}{\vlhy{}}{\vlhy{\quad }}
										}
									}{
										\vlin{\lefrul \corel}{}{\alpha^n (\vec x), x^n \Uparrow y \seqar\vec y, y }{
											\vlin{\id}{}{\alpha^n (\vec x) , x^n ,y \seqar\vec y, y }{\vlhy{}}
										}
									}
								}
							}
						}
					}{
						\vlin{\rigrul{\corel}}{}{\alpha^n (\vec x) , \qltrans \phi \plimp y \seqar \vec y \Uparrow \lneg x^n \limp y }{\vlhy{\vdots}}
					}
				}
			}
		}
	}
	\]
	\caption{Proof of $\forall$ case for left-right direction of Lemma~\ref{lem:qltrans-correct}.}
	\label{fig:ql-der-forall}
\end{figure}

\begin{proof}
	We proceed by induction on the number of quantifiers in $\phi$.
	For the base case, when $\phi $ is quantifier-free, we appeal to Prop.~\ref{prop:pos-neg-sat}. The left-right direction follows directly by weakening (cf.~Prop.~\ref{prop:wk-admiss-mall}), while the right-left direction follows after observing that $\vec y$ does not occur in $\qltrans{\phi}$ or $\alpha^n(\vec x)$; thus $\vec y $ may be deleted from a proof (along with its descendants) while preserving correctness.
	
	For the inductive step, in the left-right direction we give appropriate bi-focussed proofs in Figs.~\ref{fig:ql-der-exists} and \ref{fig:ql-der-forall}, where: $\pm x$ in Fig.~\ref{fig:ql-der-exists} is chosen to be $x$ if $x\in \alpha$ and $\lneg x $ otherwise; the derivations marked $\IH$ are obtained by the inductive hypothesis; and the derivation marked $\dots $ in Fig.~\ref{fig:ql-der-forall} is analogous to the one on the left of it.\footnote{Note that, for the derivations for the innermost quantifier ($\exists$ or $\forall$), the topmost $\rel$ or $\corel$ step of Figs.~\ref{fig:ql-der-exists} or \ref{fig:ql-der-forall} (resp.) does not occur.}
	
	For the right-left direction, we need only consider the other possibilities that could occur during bi-focussed proof search, by the focussing theorem, Thm.~\ref{thm:bi-foc-complete}.
	For the $\exists $ case, bottom-up, one could have chosen to first decide on $\qltrans \phi \limp y$ in the antecedent.
	The associated $\lefrul \limp $ step would have to send the formula $(x^n\limp y) \vlor (\lneg x^n \limp y)$ to the right premiss (for $y$), since otherwise every variable occurrence in that premiss would be distinct and there would be no way to correctly finish proof search. Thus, possibly after weakening, we may apply the inductive hypothesis to the left premiss (for $\qltrans \phi$).
	A similar analysis of the upper $\lefrul \limp$ step in Fig.~\ref{fig:ql-der-exists} means that any other split will allow us to appeal to the inductive hypothesis after weakening.
	For the $\forall $ case the argument is much simpler, since no matter which order we `co-decide', we will end up with the same leaves. (This is actually exemplary of the more general phenomenon that invertible phases of rules are `confluent', cf.~\cite{Andreoli92,ChaudhuriMS08,liang2009focusing}.)
	In particular, $\vlan$-steps may be permuted as follows:
	\[
	\begin{array}{rc}
	& \vlderivation{
		\vliin{\vlan}{}{\seqar \Gamma, A\vlan B, C\vlan D}{
			\vliin{\vlan}{}{\seqar \Gamma, A \vlan B, C }{\vlhy{\seqar \Gamma, A, C}}{\vlhy{\seqar \Gamma, B, C}}
		}{
			\vliin{\vlan}{}{\seqar \Gamma , A \vlan B, D}{\vlhy{\seqar \Gamma, A, D}}{\vlhy{\seqar \Gamma, B, D}}
		}
	}
	\\
	\noalign{\medskip} 
	\leadsto\quad  & 
	\vlderivation{
		\vliin{\vlan}{}{\seqar \Gamma, A\vlan B, C\vlan D}{
			\vliin{\vlan}{}{\seqar \Gamma, A , C \vlan D}{\vlhy{\seqar \Gamma, A, C}}{\vlhy{\seqar \Gamma, A,D}}
		}{
			\vliin{\vlan}{}{\seqar \Gamma , B, C \vlan D}{\vlhy{\seqar \Gamma, B,C}}{\vlhy{\seqar \Gamma, B, D}}
		}
	} 
	\end{array}
	\qedhere
	\]
\end{proof}

\begin{theorem}
	\label{thm:qltrans-correct}
	A closed QBF $\phi$ is true if and only if $\aMALL$ proves $\qltrans{\phi}$.
\end{theorem}
\begin{proof}
	Follows immediately from Lemma~\ref{lem:qltrans-correct}, setting $\vec y = \emptyset$.
\end{proof}

%
%
%

\section{Focussed proof search as alternating time predicates}
\label{sect:proof-search-predicates}

%
%
%
%

In this section we show how to express focussed proof search as an alternating polynomial-time predicate that will later allow us to calibrate the complexity of proof search with levels of the QBF and polynomial hierarchies.
The notions we develop apply equally to either $\MALL$ or $\aMALL$.

This section mostly follows Sect.~5 of \cite{Das18:ijcar} except that, as well as further general details, we include a proof of Theorem~\ref{thm:prov-preds-qbf} (essentially Theorem~20 in \cite{Das18:ijcar}).

\smallskip

We will now introduce `provability predicates' that delineate the complexity of proof search in a similar way to the QBF and polynomial hierarchies we presented earlier.
Recall the notions of deterministic, nondeterministic and co-nondeterministic rules from Dfn.~\ref{dfn:multifoc-mallw}, cf.~Fig.~\ref{fig:focmall}.

\begin{definition}
	[Focussing hierarchy]
	\label{dfn:foc-hier}
	A cedent $\Gamma$ of $\MALL(\mathsf w)$ is:
	\begin{itemize}
		\item 
		$\Sigma^f_0$-provable, equivalently $\Pi^f_0$-provable, if $\seqar \Gamma$ is provable using only deterministic rules.
		\item 
		{$\Sigma^f_{k+1}$-provable} if there is a derivation of $\seqar \Gamma$, using only deterministic and nondeterministic rules, from sequents $\seqar \Gamma_i$ which are $\Pi^f_k$-provable.
		\item 
		{$\Pi^f_{k+1}$-provable} if every maximal path from $\seqar \Gamma$, bottom-up, through deterministic and co-nondeterministic rules ends at a $\Sigma^f_k$-provable sequent.
	\end{itemize}
	We sometimes simply say ``$\Gamma$ is $\Sigma^f_k$'' or even ``$\Gamma \in \Sigma^f_k$'' if $\Gamma$ is $\Sigma^f_k$-provable.
\end{definition}

\noindent
The definition above is robust under the choice of multi-focussed, (co-)focussed or bi-focussed proof systems: while the number of $\dec$ or $\codec$ steps may increase, the number of alternations of nondeterministic and co-nondeterministic phases is the same.
This robustness will also apply to other concepts introduced in this section.

From the definition it is not hard to see that we have a natural correspondence between the focussing hierarchy and the other hierarchies we have discussed:

\begin{theorem}
	\label{thm:prov-preds-qbf}
	For $k\geq 0$, we have the following:
	\begin{enumerate}
		\item\label{item:pi-f-k-in-pi-p-k} $\Pi^f_k$-provability is computable in $\Pi^p_k$.
		\item\label{item:sigma-f-k-in-sigma-p-k} $\Sigma^f_k$-provability is computable in $\Sigma^p_k$.
	\end{enumerate} 
\end{theorem}

\noindent
An analogous result has been presented in previous work, \cite{Das:17:focussing}. An interesting point is that, for the $\vlte$ rule, even though there are two premisses, the rule is context-splitting, and so a nondeterministic machine may simply split into two parallel threads with no blowup in complexity.
\begin{proof}
	[of Thm.~\ref{thm:prov-preds-qbf}]
	We proceed by induction on $k$.
	In the base case, a cedent is $\Sigma^f_0$-provable (equivalently $\Pi^f_0$-provable) just if it has a proof using only deterministic rules. 
	Upon inspection of Fig.~\ref{fig:focmall}, we notice that it does not matter which order we apply deterministic rules, bottom-up, since maximal application will always lead to the same sequent at the top.
	This follows from a simple rule permutation argument:
	\[
	\vlderivation{
		\vlin{\vlpa}{}{\seqar\Gamma, A\vlpa B, C\vlpa D}{
			\vlin{\vlpa}{}{\seqar\Gamma, A \vlpa B,C,D}{\vlhy{\seqar\Gamma, A,B,C,D}}
		}
	}
	\qquad \leadsto \qquad
	\vlderivation{
		\vlin{\vlpa}{}{\seqar\Gamma, A\vlpa B, C\vlpa D}{
			\vlin{\vlpa}{}{\seqar\Gamma, A , B,C\vlpa D}{\vlhy{\seqar\Gamma, A,B,C,D}}
		}
	}
	\]
	Thus, since deterministic rules must terminate after a linear number of steps (bottom-up), we may verify deterministic-provability by simply applying deterministic steps maximally (in any order bottom-up) and verifying that the end result is a correct initial sequent. Thus we indeed have that $\Sigma^f_0$-provability (equivalently $\Pi^f_0$-provability) is computable in $\ptime = \Sigma^p_0 = \Pi^p_0$.

	For the inductive step for \ref{item:pi-f-k-in-pi-p-k}, a cedent $\Gamma$ is not $\Pi^f_{k+1}$-provable just if there is some maximal branch of co-nondeterministic steps applied to $\Gamma$, bottom-up, that terminates in a sequent $\Gamma'$ that is not $\Sigma^f_k$-provable. 
	Any such branch has polynomial-size (by inspection of the rules), and we have from the inductive hypothesis that $\Sigma^f_k$-provability is computable in $\Sigma^p_k$. 
	Thus we have that non-$\Pi^f_{k+1}$-provability is computable in $\Sigma^p_{k+1}$ and so, by Dfn.~\ref{dfn:ph}, $\Pi^f_{k+1}$-provability is computable in $\Pi^p_{k+1}$.
	

	For the inductive step for \ref{item:sigma-f-k-in-sigma-p-k}, a cedent $\Gamma$ is $\Sigma^f_{k+1}$-provable just if there is a derivation of the form,
	\[
	\vltreeder{\Phi}{\seqar \Gamma}
	{\seqar \Gamma_1}{\cdots}{\seqar \Gamma_n}
	\]
	using only deterministic and nondeterministic rules such that each $\Gamma_i$ is $\Pi^f_k$-provable.
	Notice that such $\Phi$ has polynomial-size in $|\Gamma|$ with $\sum_{i=1}^n |\Gamma_i| \leq |\Gamma|$, by inspection of the (non)determinstic rules in Fig.~\ref{fig:focmall}.
	Thus, to check that $\Gamma$ is $\Sigma^f_{k+1}$-provable we need only guess the appropriate derivation $\Phi$ and sequents $\Gamma_1, \dots , \Gamma_n$, and then check that each $\Gamma_i$ is $\Pi^f_k$-provable.
	By the inductive hypothesis, we have that $\Pi^f_k$-provability is a $\Pi^p_k$ property,
	so
	we may check on a $\Pi_k$-machine that \emph{all} these $\Gamma_i$s are actually $\Pi^f_k$-provable in time $\sum_{i=1}^n |\Gamma_i|^{O(1)} \leq |\Gamma|^{O(1)}$.
	Thus we have that $\Sigma^f_{k+1}$-provability is indeed computable in $\Sigma^p_k$.	
	%
\end{proof}

\begin{corollary}
	\label{cor:prov-pred-qbfs}
	For $k\geq 1$, we have the following:
	\begin{enumerate}
		\item There are $\Sigma^q_k$ formulas $ \Prov{\Sigma^f_k}_n$, constructible in time polynomial in $n\in \Nat$, computing $\Sigma^f_k$-provability on all formulas $A$ s.t.\ $|A|=n$.
		\item There are $\Pi^q_k$ formulas $\Prov{\Pi^f_k}_n$, constructible in time polynomial in $n\in \Nat$, computing $\Pi^f_k$-provability on formulas $A$ s.t.\ $|A|=n$.
	\end{enumerate} 
\end{corollary}
\begin{proof}
	Follows immediately from Theorem~\ref{thm:prov-preds-qbf} under Theorem~\ref{thm:qbf-complexity}.
\end{proof}

We now give a slightly different way to view the focussing hierarchy, based on a more directly calculable measure of cedents that is similar to the notions of `decide depth' and `release depth' found in other works, e.g.~\cite{nigaminvestigating}.
We will use (a variation of) this to eventually formulate our encoding from $\aMALL$ to $\qcpl$ in the next section.

\begin{definition}
	[(Co-)nondeterministic complexity]
	Let $\Phi$ be a $\focMALL(\mathsf w)$ proof.
	We define the following:
	\begin{itemize}
		\item The \emph{nondeterministic} complexity of $\Phi$, written $\sigma(\Phi)$, is the maximum number of alternations, bottom-up, between $\dec$ and $\codec$ steps in a branch through $\Phi$, setting $\sigma(\Phi) = 1 $ if $\Phi$ has only $\dec$ steps.
		\item The \emph{co-nondeterministic} complexity of $\Phi$, written $\pi(\Phi)$, maximum number of alternations, bottom-up, between $\dec$ and $\codec$ steps in a branch through $\Phi$, setting $\pi(\Phi) = 1$ if $\Phi$ has only $\codec$ steps.
	\end{itemize}
	
	%
	
	
	\noindent
	For a cedent $\Gamma$ we further define the following:
	\begin{itemize}
		\item $\sigma( \Gamma)$ is the least $k\in \Nat$ s.t.\ there is a $\focMALL(\mathsf w)$ proof $\Phi $ of $\seqar \Gamma$ with $\sigma(\Phi)=  k$.
		\item $\pi(\Gamma)$) is the least $k\in \Nat$ s.t.\ there is a $\focMALL(\mathsf w)$ proof $\Phi $ of $\seqar \Gamma$ with $\pi(\Phi) = k$.
	\end{itemize} 
\end{definition}

%


\noindent
Putting together the results and notions of this section, we have:

\begin{proposition}
	\label{prop:complexity-foc-hier}
	Let $\Gamma$ be a cedent and $k \geq 0$. We have the following:
	\begin{enumerate}
		\item $\Gamma $ is $\Sigma^f_k$-provable if and only if $\sigma (\Gamma)\leq k$.
		\item $\Gamma$ is $\Pi^f_k$-provable if and only if $\pi(\Gamma)\leq k$. 
	\end{enumerate}
\end{proposition}


%
%

\section{An `inverse' encoding from $\aMALL$ into $\qcpl$}
\label{sect:lqtrans}

In this section we will use the ideas of the previous section to give an explicit encoding from $\aMALL$ to $\qcpl$, i.e.\ a polynomial-time mapping from $\aMALL$-formulas to QBFs whose restriction to theorems has image in $\qcpl$.
Moreover, we will show that this encoding acts as an `inverse' to the one we gave in Sect.~\ref{sect:qltrans}, and finally identify natural fragments of $\aMALL$ complete for each level of $\pth$.
%

%
%
%
%
%

This section mostly follows Sect.~6 of \cite{Das18:ijcar}, except that we give significantly more proof details.

\subsection{Approximating (co-)nondeterministic complexity}

The nondeterministic and co-nondeterministic complexities $\sigma$ and $\pi$ we introduced in the previous section do not give us a bona fide encoding from $\aMALL$ to true QBFs since 
%
%
they are hard to compute. Instead we give an `overestimate' here that will suffice for the encodings we are after.
This overestimate will be parametrised by some enumeration of all formulas,\footnote{Strictly speaking, we really mean `formula occurrences' rather than just `formulas', but we will sweep this technicality under the carpet in the interest of a lighter exposition.} which will drive the possible choices during proof search.
However we will later show that the choice of this enumeration is irrelevant, meaning that the approximation can be flexibly calculated and is in fact polynomial-time computable.

Another option might have been, rather than taking the `least' formula in a sequent under some enumeration, that we may view the sequent as a list instead in a calculus with an exchange rule.
We avoided this in the interest of having a terminating proof system.

Throughout, we will identify enumerations with total orders in the natural way.


\begin{definition}
	[Approximating the complexity of a sequent]
	\label{dfn:ndcomp-condcomp}
	Let $\prec$ be a total order on all $\MALL(\mathsf w)$ formulas.
	We define the functions $\ndcomp_\prec$ and $\condcomp_\prec$ on sequents in Fig.~\ref{fig:approx-complex}.\footnote{In the conference version, \cite{Das18:ijcar}, there was an error in the base case, where $0$ was written instead of $1$.
	}
	\begin{figure}[t]
		\[
		\begin{array}{rcll}
		\ndcomp_\prec (\vec a) & \ \dfn\  & 1 & \\
		\ndcomp_\prec (\Gamma, A \vlpa B) & \ \dfn \ & \ndcomp_\prec (\Gamma, A , B) &  
		\\
		\ndcomp_\prec (\vec a , \vec P , P) & \ \dfn \ & \ndcomp_\prec (\vec a , \vec P \Downarrow P) &  \text{$P$ is $\prec$-least in $\vec P , P$} \\
		\ndcomp_\prec (\vec a , \vec P , \vec M ,M ) & \ \dfn \ & 1 + \condcomp_\prec (\vec a , \vec P , \vec M,M ) & 
		\\
		\noalign{\medskip}
		\condcomp_\prec (\vec a) & \ \dfn\  & 1 & \\
		\condcomp_\prec (\Gamma, A \vlpa B) & \ \dfn \ & \condcomp_\prec (\Gamma, A , B) &   
		\\
		\condcomp_\prec (\vec a , \vec P , P) & \ \dfn \ & 1 +\ndcomp_\prec (\vec a , \vec P , P) &  
		\\
		\condcomp_\prec (\vec a , \vec P , \vec M , M ) & \ \dfn \ &  \condcomp_\prec (\vec a , \vec P , \vec M  \Uparrow M) &  \text{$M$ is $\prec$-least in $\vec M , M$}
		\\
		\noalign{\bigskip}
		\ndcomp_\prec (\Gamma \Downarrow A \vlor B) & \ \dfn \ & \begin{cases}
		\ndcomp_\prec (\Gamma, A) & \ndcomp_\prec (A) \geq \ndcomp_\prec (B) \\
		\ndcomp_\prec (\Gamma, B ) & \text{otherwise}
		\end{cases} & \\
		\noalign{\smallskip}
		\ndcomp_\prec (\Gamma \Downarrow A \vlte B ) & \ \dfn \ & \begin{cases}
		\ndcomp_\prec (\Gamma, A) & \ndcomp_\prec (A) \geq \ndcomp_\prec (B) \\
		\ndcomp_\prec (\Gamma, B ) & \text{otherwise}
		\end{cases} & \\
		\noalign{\smallskip}
		\ndcomp_\prec (\Gamma \Downarrow X) & \ \dfn \ & \ndcomp_\prec (\Gamma, X) & \text{$X$ is $a$ or $N$} \\
		\noalign{\bigskip}
		\condcomp_\prec (\Gamma \Uparrow A \vlan B) &\ \dfn\ & \begin{cases}
		\condcomp_\prec (\Gamma, A) & \condcomp_\prec (A) \geq \condcomp_\prec (B) \\
		\condcomp_\prec (\Gamma, B ) & \text{otherwise}
		\end{cases} & \\
		\noalign{\smallskip}
		\condcomp_\prec (\Gamma \Uparrow X) & \ \dfn\ & \condcomp_\prec (\Gamma, X) & \text{$X$ is $O$ or $P$}
		\end{array}
		\]
		\caption{Approximating (co-)nondeterminstic complexities.}
		\label{fig:approx-complex}
	\end{figure}	
\end{definition}





\begin{proposition}
	[Confluence]
	\label{prop:confluence}
	For any two total orders $\prec$ and $\prec'$ on $\MALL(\mathsf w)$ formulas, we have that $\ndcomp_\prec (\Gamma) = \ndcomp_{\prec'} (\Gamma)$ and $\condcomp_\prec (\Gamma) = \condcomp_{\prec'} (\Gamma)$.
	
\end{proposition}


\noindent
To prove this, we give essentially a confluence argument for terminating relations, but we avoid using a formal rewriting argument for self-containedness.

\begin{proof}[of Prop.~\ref{prop:confluence}]
	%
	We proceed by induction on the number of connectives in $\Gamma$.
	The base case, when $\Gamma$ consists of only atomic formulas, is trivial, so we consider the inductive steps.
	When invoking the inductive hypothesis, we may freely suppress the subscripts $\prec$ or $\prec'$. 
	
	Suppose that, at some point along the definition of the approximations, $\prec $ and $\prec'$ disagree on what the least formula is. Namely, the $\prec$-least formula is $P_0$ and the $\prec'$-least formula is $P_1$.
	In this case we have the following situation:
	\[
	\begin{array}{rcll}
	\ndcomp_\prec (\vec a, \vec P, P_0, P_1) & = & \ndcomp_\prec (\vec a, \vec P , P_0 \Downarrow P_1) & \\
	& \vdots & &  \\
	& = & \delta_1 + \ndcomp_\prec (\vec a, \vec P, P_0, \vec P_1)  &  \\
	& = & \delta_1 + \ndcomp (\vec a, \vec P,  \vec P_1 \Downarrow P_0) & \text{by inductive hypothesis} \\
	& \vdots & &  \\
	& = &  \delta_0 + \delta_1 + \ndcomp (\vec a, \vec P, \vec P_1, \vec P_0) & 
	\end{array}
	\]
	where each $\delta_i$ is either $2$ or $0 $ depending on whether a co-nondeterministic phase is entered or not during the bi-pole induced by $P_i$.
	We will have a similar derivation for $\prec'$, with only $\delta_0$ and $\delta_1$ swapped, whence we conclude by commutativity and associativity of addition.
	
	The case where a $M$-formula is chosen in the definition of $\condcomp$ is similar.

\end{proof}


\noindent
From now on, we may suppress the subscript $\prec$ for the notions $\ndcomp$ and $\condcomp$.

\begin{corollary}
	[Overapproximation]
	\label{cor:over-estimate}
	$\sigma(\Gamma)\leq \ndcomp (\Gamma)$ and $\pi(\Gamma)\leq \condcomp(\Gamma)$.
\end{corollary}
\begin{proof}
	A proof with the minimal number of alternations between nondeterministic and co-nondeterministic phases will induce a strategy $\prec$ on which we may evaluate $\ndcomp$ and $\condcomp$.
	These will be bounded below by their actual $\sigma $ and $\pi$ values.
\end{proof}
Notice that the over-estimation for the $\vlte$ case is particularly extreme: in the worst case we have that the entire context is copied to one branch. 
In fact we could optimise this somewhat, by only considering `plausible' splittings, but it will not be necessary for our purposes.




\begin{corollary}
	[Feasibility]
	\label{cor:approx-ptime-computable}
	$\ndcomp$ and $\condcomp$ are polynomial-time computable.
\end{corollary}
\begin{proof}
	Clearly, $\ndcomp_\prec$ and $\condcomp_\prec$ are polynomial-time computable for any polynomial-time enumeration $\prec$.
	So we may simply pick any polynomial-time enumeration of the formulas and appeal to Prop.~\ref{prop:confluence}. 
\end{proof}

\subsection{Tightness of approximations in the image of $\qltrans{\cdot}$}
Since $\qltrans{\phi}$ is always a relatively `balanced' formula, we have that the overestimation just defined is in fact tight in the image of $\qltrans{\cdot}$ from $\aMALL$:
\begin{proposition}
	[Tightness]
	\label{prop:tight-approx}
	For $k \geq 1$ we have the following:
	\begin{enumerate}
		\item If $\phi \in \Sigma^q_k \setminus \Pi^q_k$ then $\ndcomp (\qltrans{\phi}) = \sigma (\qltrans \phi) = k$ and $\condcomp (\qltrans \phi) = \pi(\qltrans \phi) = 1+k$.
		\item If $\psi \in \Pi^q_k \setminus \Sigma^q_{k}$ then $\condcomp (\qltrans{\psi}) = \pi (\qltrans \phi) = k$
		and $\ndcomp (\qltrans \psi) = \pi(\qltrans \psi) = 1+k$.
	\end{enumerate}

\end{proposition}

\noindent
Intuitively, the tightness of the approximation follows from the following two properties of the derivations in Figs.~\ref{fig:ql-der-exists} and \ref{fig:ql-der-forall}:
\begin{itemize}
	\item There is only one non-atomic formula per sequent, so the $\vlte$-overapproximation is not significant.
	\item Weakening is only required on atomic formulas, so initial sequents need not be further broken down in the definitions of $\ndcomp$ and $\condcomp$.
\end{itemize} 

\begin{proof}[of Prop.~\ref{prop:tight-approx}]
	We have that $\sigma (\phi) = k $ and $\pi(\psi)= k$ already from the proof of Lemma~\ref{lem:qltrans-correct}, so it remains to show that the approximations are tight.
	For this we show by induction on the number of quantifiers in $\phi$ or $\psi$ that, more generally:
	\begin{itemize}
		\item $\ndcomp (\qltrans \phi, \vec a) = k$ for any sequence $\vec a$ of atomic formulas.
		\item $\condcomp (\qltrans \psi, \vec a) = k$ for any sequence $\vec a$ of atomic formulas.
	\end{itemize}
	
	When $\phi $ and $\psi $ are quantifier-free, notice that $\ndcomp (\phi^+, \vec a) = 1 = \condcomp (\phi^-, \vec a)$, since $\phi^+$ has only positive connectives and $\phi^-$ has only negative connectives.

	If $\phi$ is $\exists x. \phi'$ then $\ndcomp (\qltrans {\phi'},\vec b) = k $, for any $\vec b$, by the inductive hypothesis so:
	\[
	\begin{array}{rcll}
	\ndcomp (\qltrans \phi , \vec a) & = & \ndcomp ((\qltrans{\phi'} \limp y ) \limp ( (x^n \limp y ) \vlor (\cnot x^n \limp y ) ) ,\vec a)  & \text{by definition of $\qltrans{\cdot}$} \\
	& = & \ndcomp (\qltrans{\phi'} \vlte \cnot y ,  (\cnot x^n \vlpa y) \vlor (x^n \vlpa  y) ,\vec a) & \text{by definition of $\limp$} \\
	& = & \ndcomp (\qltrans{\phi'} \vlte \cnot y , x^n \vlpa y ,\vec a) & \text{by Prop.~\ref{prop:confluence}} \\
	& = & \ndcomp (\qltrans{\phi'} \vlte y , x^n , y , \vec a) & \\
	& = & \ndcomp (\qltrans{\phi'}, x^n, y, \vec a) & \text{since $\ndcomp (y)  = 1$} \\
	& = & k & \text{by inductive hypothesis}
	\end{array}
	\]
	We also have that $\condcomp (\qltrans{\phi}, \vec a) = 1+ \ndcomp (\qltrans{\phi'}, x^n, y, \vec a) = 1+k$, by a similar analysis.
	
	If $\psi $ is $\forall x . \psi'$ then $\condcomp (\qltrans{\psi'}, \vec b) = k$, for any $\vec b$, by the inductive hypothesis so:
	\[
	\begin{array}{rcll}
	\condcomp (\qltrans \psi , \vec a) & = & \condcomp ((\qltrans{\psi'} \plimp y ) \limp ( (x^n \limp y ) \vlan (\cnot x^n \limp y ) ) ,\vec a)  & \text{by definition of $\qltrans{\cdot}$} \\
	& = & \condcomp (\qltrans{\psi'} \vlan \cnot y ,  (\cnot x^n \vlpa y) \vlan (x^n \vlpa  y) ,\vec a) & \text{by definition of $\limp, \plimp$} \\
	& = & \condcomp (\qltrans{\psi'} \vlan \cnot y , x^n \vlpa y ,\vec a) & \text{by Prop.~\ref{prop:confluence}} \\
	& = & \condcomp (\qltrans{\psi'} \vlan y , x^n , y , \vec a) & \\
	& = & \condcomp (\qltrans{\psi'}, x^n, y, \vec a) & \text{since $\condcomp (y) = 1$} \\
	& = & k & \text{by inductive hypothesis}
	\end{array}
	\]
	We also have that $\ndcomp (\qltrans{\psi}, \vec a) = 1+\condcomp (\qltrans{\psi'}, x^n, y, \vec a) =  1 +k$, by a similar analysis. 
\end{proof}


\subsection{An encoding from $\aMALL$ to QBFs and main results}
From Cor.~\ref{cor:prov-pred-qbfs}, let us henceforth fix appropriate QBFs $\Prov{\Sigma^f_k}_n$ and $\Prov{\Pi^f_k}_n$, for $k\geq 1$, computing $\Sigma^f_k$-provability and $\Pi^f_k$-provability in $\afocMALL$, respectively, for formulas of size $n$.
We are now ready to define our `inverse' encoding of $\qltrans{\cdot}$:

\begin{definition}
	[$\aMALL$ to $\qcpl$]
	For a $\aMALL$ formula $A$, we define:
	\[
	\lqtrans{A}
	\ \dfn  \ 
	\begin{cases}
	\Prov{\Sigma^f_{k}}_{|A|}(A) & \text{if }k=\ndcomp (A)\leq \condcomp (A) \\
	\Prov{\Pi^f_{k}}_{|A|}(A) & \text{if }k=\condcomp(A) < \ndcomp(A)
	\end{cases}
	\]
\end{definition}

\noindent
Finally, we are able to present our main result:

\begin{theorem}
	\label{thm:main}
	We have the following:
	\begin{enumerate}
		\item\label{item:qltrans-correct} $\qltrans\cdot$ is a polynomial-time encoding from $\qcpl$ to $\aMALL$.
		\item\label{item:lqtrans-correct} $\lqtrans{\cdot}$ is a polynomial-time encoding from $\aMALL$ to $\qcpl$.
		\item\label{item:inverse-encodings} The composition $ \lqtrans{\cdot} \circ \qltrans{\cdot} : \qcpl \to \qcpl$ preserves quantifier complexity, i.e., for $k \geq 1$, it maps true $\Sigma^q_k$ ($\Pi^q_k$) sentences to true $\Sigma^q_k$ (resp.~$\Pi^q_k$) sentences.
	\end{enumerate}
	
\end{theorem}
\begin{proof}
	We already proved \ref{item:qltrans-correct} in Thm.~\ref{thm:qltrans-correct}.
	\ref{item:lqtrans-correct} follows from the definitions of $\Prov{\Sigma^f_k}$ and $\Prov{\Pi^f_k}$ (cf.~Cor.~\ref{cor:prov-pred-qbfs}), under Prop.~\ref{prop:complexity-foc-hier} and Cors.~\ref{cor:over-estimate} and \ref{cor:approx-ptime-computable}.
	Finally \ref{item:inverse-encodings} then follows by tightness of the approximations $\ndcomp$, $\condcomp$ in the image of $\qltrans{\cdot}$, Prop.~\ref{prop:tight-approx}.
\end{proof}

\noindent
Consequently, we may identify polynomial-time recognisable subsets of $\aMALL$-formulas whose theorems are complete for levels of the polynomial hierarchy:

\begin{corollary}
	We have the following, for $k\geq 1$:
	\begin{enumerate}
		\item $\{ A :  \text{$\ndcomp (A) \leq k$ and $\aMALL$ proves $A$} \}$ is $\Sigma^p_k$-complete.
		\item $\{ A : \text{$\condcomp (A) \leq k $ and $\aMALL$ proves $A$} \}$ is $\Pi^p_k$-complete.
	\end{enumerate}
\end{corollary}

\section{Extending the approach to (non-affine) $\MALL$}
\label{sect:mallph}

It is natural to wonder whether a similar result to Thm.~\ref{thm:main} could be obtained for $\MALL$, i.e.\ without weakening.
The reason we chose $\aMALL$ is that it allows for a robust and uniform approach that highlights the capacity of focussed systems to obtain tight alternating time bounds for logics, without too many extraneous technicalities.
However, the same approach does indeed extend to $\MALL$ with only local adaptations.
We give the argument in this section.

This section is comprised of new material not present in \cite{Das18:ijcar}.


\subsection{Encoding weakening in $\MALL$}

There is a well-known embedding of $\aMALL$ into $\MALL$ by recursively replacing every subformula $A$ by $\bot \vlor A$.
However, doing this might considerably increase the alternation complexity of proof search, adding up to one alternation per subformula. 
Instead, we notice that we need only conduct this replacement on literals, since those are the only ones that are weakened in the proofs of Sect.~\ref{sect:qltrans}.
From here we realise that the consideration of formulae of the form $\bot \vlor a$ (or variants thereof) may be delayed to the end of proof search.
To formalise this appropriately, we first need the following notion:
\begin{definition}
	[Weakened formulas]
	Let $\Phi$ be a $\aMALL$ proof whose initial $\id$-sequents are $\{ \Gamma_i, a_i, \cnot a_i \}_{i<m}$
	and initial $1$-sequents are $\{\Gamma_i, 1 \}_{m\leq i<n} $
	The \emph{weakened formulas} of $\Phi$ is the set $\Omega \dfn \bigcup_{i<n}  \Gamma_i$. 
\end{definition}

\noindent
We identify elements of $\Omega $ in the above definition with subformula occurrences of the conclusion of $\Phi$ in the natural way.
We have the following folklore result:

\begin{lemma}
	[Weakening lemma]
	\label{lem:weakening-lemma}
	We have the following:
	\begin{enumerate}
		\item\label{item:mall-proves-A-imp-botorA} $\MALL$ proves $A \seqar \bot \vlor A$.
		\item\label{item:amall-proves-botorA-imp-A} $\aMALL$ proves $\bot \vlor A \seqar A$, with weakening only on $A$. 
		\item\label{item:weakened-formulas-bot} Let $A $ be a $\MALL(\mathsf w)$ formula and $\Omega$ a set of subformula occurrences of $A$. 
		There is a $\aMALL$ proof $\Phi$ of $A$ with weakened formulas among $\Omega$ if and only if $\MALL$ proves $A[\bot \vlor B/B]_{B \in \Omega}$.
	\end{enumerate}
\end{lemma}

\begin{proof}
	\ref{item:mall-proves-A-imp-botorA} and \ref{item:amall-proves-botorA-imp-A} are given by the following derivations:
	\[
	\vlderivation{
		\vlin{}{}{A \seqar \bot \vlor A}{
			\vliq{}{}{A \seqar A}{\vlhy{}}
		}
	}
	\qquad
	\vlderivation{
		\vliin{}{}{\bot \vlor A \seqar A}{
			\vlin{}{}{\bot \seqar A}{\vlhy{}}
		}{
			\vliq{}{}{A \seqar A}{\vlhy{}}
		}
	}
	\]
	\ref{item:weakened-formulas-bot} now follows immediately from \ref{item:mall-proves-A-imp-botorA} and \ref{item:amall-proves-botorA-imp-A} under the `deep inference' property:
	\begin{equation}
	\text{\emph{If $\MALL(\mathsf w)$ proves $A(B)$ and $B \seqar C$ then it proves $A(C)$.}}
	\end{equation}
	This property is well known (see, e.g., \cite{Stra02:ll-di}) and follows by a routine induction on the structure of $A$, in particular appealing to the cut-elimination property of $\MALL(\mathsf w)$.
	For instance here are the cases when $A$ is a $\vlpa$ or $\vlor$ formula,
	\[
	\vlderivation{
		\vliin{\cut}{}{\seqar A_0(C) \vlpa A_1(C)}{
			\vliq{}{}{{\seqar A_0 (B) \vlpa A_1 (B)}}{\vlhy{}}
		}{
			\vlin{\rigrul{\vlpa}}{}{A_0 (B) \vlpa A_1(B) \seqar A_0(C) \vlpa A_1(C) }{
				\vliin{\lefrul \vlpa}{}{A_0 (B) \vlpa A_1(B) \seqar A_0(C) , A_1(C)}{
					\vltr{\IH}{A_0 (B)  \seqar A_0(C)}{\vlhy{\ }}{\vlhy{}}{\vlhy{\ }}
				}{
					\vltr{\IH}{A_1(B) \seqar A_1(C)}{\vlhy{\ }}{\vlhy{}}{\vlhy{\ }}
				}
			}
		}
	}
	\]
	\[
	\vlderivation{
		\vliin{\cut}{}{\seqar A_0 (C) \vlor A_1 (C) }{
			\vliq{}{}{\seqar A_0 (B) \vlor A_1 (B)}{\vlhy{}}
		}{
			\vliin{\lefrul \vlor}{}{ A_0 (B) \vlor A_1 (B)  \seqar A_0 (C) \vlor A_1 (C)  }{
				\vlin{\rigrul \vlor}{}{ A_0 (B) \seqar A_0 (C) \vlor A_1 (C)  }{
					\vltr{\IH}{A_0 (B)  \seqar A_0(C)}{\vlhy{\ }}{\vlhy{}}{\vlhy{\ }}
				}
			}{
				\vlin{\rigrul \vlor}{}{A_1 (B) \seqar A_0 (C) \vlor A_1 (C) }{
					\vltr{\IH}{A_1(B) \seqar A_1(C)}{\vlhy{\ }}{\vlhy{}}{\vlhy{\ }}
				}
			}
		}
	}
	\]
	where the derivations marked $\IH$ are obtained from the inductive hypothesis.
	The cases when $A$ is a $\vlte$ or $\vlan$ formula are similar to the two cases above.
\end{proof}

\subsection{Adapting the translation $\qltrans{\cdot}$ for $\MALL$}


%
It turns out that, in the arguments of Sect.~\ref{sect:qltrans}, we only used weakenings on atomic formulae: notice that, in the proofs of Prop.~\ref{prop:pos-neg-sat} and of Lemma~\ref{lem:qltrans-correct}, the only weakenings were applied on literals, and all initial sequents had the form $\vec a$. 

\begin{observation}
	\label{obs:atomic-weakening}
	The proof of Thm.~\ref{thm:qltrans-correct} requires weakening only on literals.
\end{observation}

\noindent
This motivates the following definition:

\begin{definition}
	For a $\MALL(\mathsf{w})$ formula $A$, write $A' $ for the result of replacing every literal occurrence $a$ with $\bot \vlor a$.
\end{definition}

\noindent
We now have the following immediately from Lemma~\ref{lem:weakening-lemma}:
\begin{proposition}
	$\aMALL$ proves $\qltrans \phi$ if and only if $\MALL$ proves $\qltrans \phi' $.
\end{proposition}
%
%
%
%
%

\subsection{Dealing with $\bot \vlor a$ formulas deterministically}

Even though we have restricted our treatment of weakened formulas to only literals, these may still a priori increase the alternation complexity of proof search linearly under the encoding $A'$.
To avoid this, we will work in a certain normal form of proofs that delays consideration of formulas of the form $\bot \vlor a$ until the end of bottom-up proof search.
To enforce this we must slightly `hack' the proof systems and complexity approximations previously introduced.

For generality, let us introduce new metavariables $c,d$ etc.\ varying over $\vlor$-clauses containing at least one $\bot$:
\[
c \quad ::= \quad \bot \quad | \quad c \vlor x \quad | \quad x \vlor c \quad | \quad \cnot x \vlor c \quad | \quad c \vlor \cnot x \quad | \quad c \vlor c
\]
\noindent
Our intention is to treat $c$-formulas much like atoms in proof search, in particular not decomposing them until the end.
To this end, we will introduce a new focussed system $\focMALL'$ that enforces this within the rules.

\begin{definition}
	[$\focMALL'$]
	\label{dfn:focmall'}
	We temporarily redefine the metavariables $M,N,O,P,Q$ so that only $O$ is permitted to vary over $c$-formulas,i.e.\ if $P$ or $Q$ is a $\vlor$-formula it must not be a $c$-formula.
	The proof system $\focMALL'$ is hence defined just as $\focMALL$ in Fig.~\ref{fig:focmall}, under this revision of metavariables, with the following exceptions:
	\begin{itemize}
		\item $\vec a$ in Fig.~\ref{fig:focmall} is everywhere replaced by $\vec a, \vec c$, i.e.\ $\focMALL'$ has the following rules,
		\[
		\vlinf{\dec}{}{\seqar \vec a , \vec c , \vec P , \vec P'}{\seqar \vec a , \vec c , \vec P \Downarrow \vec P'}
		\qquad
		\vlinf{\codec}{}{\seqar \vec a , \vec c , \vec P , \vec M}{\seqar \vec a , \vec c , \vec P \Uparrow \vec M}
		\qquad
		\vlinf{\rel}{}{\seqar \Gamma \Uparrow \vec a, \vec c, \vec N}{\seqar \Gamma, \vec a , \vec c, \vec N}
		\]
		instead of their analogous versions written in Fig.~\ref{fig:focmall}.
		\item $\focMALL'$ has the following additional initial sequents:
		\[
		\vlinf{}{}{\seqar \vec c , c(x), d(\cnot x)}{}
		\qquad
		\vlinf{}{}{\seqar \vec c, c(x), \cnot x}{}
		\qquad 
		\vlinf{}{}{\seqar \vec c, c(\cnot x), x}{}
		\qquad 
		\vlinf{}{}{\seqar \vec c , x , \cnot x}{}
		\qquad
		\vlinf{}{}{\seqar \vec c , 1}{}
		\]
		where literals in parentheses must occur in their respective formulas. 
		These new initial sequents are deterministic.
	\end{itemize} 
	(Co-)focussed and bi-focussed proofs of $\focMALL'$ are defined analogously to Dfn.~\ref{dfn:co-bi-focussed}.
\end{definition}


\begin{proposition}
	(Bi-focussed) $\focMALL'$ is sound and complete for $\MALL$.
\end{proposition}
\begin{proof}
	Soundness is routine, with the new initial sequents proved using simple $\vlor$ and $\bot $ steps, along with the other initial rules. 
	%
	
	For completeness, we proceed by induction on the size of a $\focMALL$ proof, by essentially a rule permutation argument.
	The critical cases are when a $\focMALL$ proof focusses on a $c$-formula, which we adapt as follows:
	\[
	\vlderivation{
		\vlin{}{}{\seqar \vec a,\vec c,  \vec P, c}{
			\vlin{}{}{
				\begin{array}{c}
				\seqar \vec a,\vec c, \vec P \Downarrow \bot \\
				\vdots \\
				\seqar \vec a,\vec c, \vec P \Downarrow c
				\end{array}
			}{
				\vlin{}{}{\seqar\vec a,\vec c,  \vec P , \bot}{
					\vltr{\Phi}{\seqar\vec a,\vec c,  \vec P}{\vlhy{\quad}}{\vlhy{}}{\vlhy{\quad}}
				}
			}
		}
	}
	\quad\leadsto\quad
	\vlderivation{
		\vltr{c,\Phi}{\seqar \vec a,\vec c, c,\vec P}{\vlhy{\quad}}{\vlhy{\quad}}{\vlhy{\quad}}
	}
	\qquad 
	\qquad 
	\vlderivation{
		\vlin{}{}{\seqar \vec a,\vec c, \vec P, c}{
			\vlin{}{}{
				\begin{array}{c}
				\seqar\vec a, \vec c, \vec P \Downarrow a \\
				\vdots \\
				\seqar\vec a,\vec c,  \vec P \Downarrow c
				\end{array}
			}{
				\vltr{\Phi}{\seqar\vec a, \vec c, \vec P, a}{\vlhy{\quad}}{\vlhy{}}{\vlhy{\quad}}
			}
		}
	}
	\quad \leadsto \quad
	\vlderivation{
		\vltr{\Phi[c/a]}{\seqar \vec a,\vec c, \vec P, c}{\vlhy{\quad}}{\vlhy{\quad}}{\vlhy{\quad}}
	}
	\]
	where,
	\begin{itemize}
		\item $c,\Phi$ is obtained from $\Phi$ by inductively adding $c$ to each premiss of an inference step in $\Phi$, bottom-up, except at $\vlte$ steps where we must only add $c$ to one premiss, say the left one. This transformation preserves the local correctness of the proof, in particular since initial sequents are closed under addition of $c$-formulas.
		\item $ a $ is a literal and $\Phi[c/a]$ is obtained from $\Phi$ by replacing every (indicated) occurrence of $a $ by $c$. Since $c$ contains $a$ as a subformula, each initial sequent transformed in this way will again be an initial sequent.
	\end{itemize}
	All other cases are routine, simply mimicking the given $\focMALL$ proof.
\end{proof}

%

%
%
%
%



\subsection{Adapting the translation $\lqtrans{\cdot}$ and main results}

Finally we adapt the approximations of (co-)nondeterministic complexity to reflect the proof search dynamics of $\focMALL'$:

\begin{definition}
	[Approximating alternating complexity in $\focMALL'$]
	$\ndcomp_\prec'$ and $\condcomp_\prec'$ are defined exactly as $\ndcomp_\prec$ and $\condcomp_\prec$ in Fig.~\ref{fig:approx-complex}, under the metavariable conventions of Dfn.~\ref{dfn:focmall'}, with the following exception:
	$\vec a$ and $a$ are replaced everywhere by $\vec a, \vec c$ and ``$a$ or $c$'', respectively. I.e.\ we have the following clauses,
	\[
	\begin{array}{rcll}
	\ndcomp_\prec' (\vec a, \vec c) & \dfn & 1 & \\
	\ndcomp_\prec' ( \vec a , \vec c, \vec P, P) & \dfn & \ndcomp_\prec'(\vec a , \vec c , \vec P \Downarrow P) & \text{$P$ is $\prec$-least in $\vec P,P$} \\
	\ndcomp_\prec' (\vec a , \vec c, \vec P, \vec M, M) & \dfn & 1 + \condcomp_\prec' (\vec a, \vec c, \vec P, \vec M, M) & \\
	\noalign{\medskip}
	\condcomp_\prec' (\vec a , \vec c) & \dfn & 1 & \\
	\condcomp_\prec' (\vec a , \vec c, \vec P, P) & \dfn & 1 + \ndcomp_\prec' (\vec a , \vec c , \vec P) \\
	\condcomp_\prec' (\vec a, \vec c, \vec P , \vec M , M) & \dfn & \condcomp_\prec' (\vec a, \vec c, \vec P , \vec M, \Uparrow M) & \text{$M$ is $\prec$-least in $\vec M, M$}\\
	\noalign{\medskip}
	\ndcomp_\prec' (\Gamma \Downarrow X ) & \dfn & \ndcomp_\prec' (\Gamma, X)  & \text{$X$ is $a$ or $c$ or $N$}
	\end{array}
	\]
	instead of their analogous versions written in Fig.~\ref{fig:approx-complex}.
\end{definition}

It is not hard to see that the results of Sect.~\ref{sect:proof-search-predicates} are applicable also to the notions $\ndcomp'$ and $\condcomp'$ developed here.
In particular Prop.~\ref{prop:confluence} holds also for $\ndcomp'$ and $\condcomp'$ and we may similarly omit the $\prec$-subscript henceforth.
All together, this allows us to define a similar encoding from $\MALL$ formulas to QBFs.

First, appealing to Cor.~\ref{cor:prov-pred-qbfs},\footnote{Notice that this holds also for the system $\focMALL'$.} let us henceforth fix appropriate QBFs $\Prov{\Sigma^f_k}_n'$ and $\Prov{\Pi^f_k}_n'$, for $k\geq 1$, computing $\Sigma^f_k$-provability and $\Pi^f_k$-provability in $\focMALL'$, respectively, for formulas of size $n$.

\begin{definition}
	[$\MALL$ to $\qcpl$]
	For a $\MALL$ formula $A$, we define:
	\[
	\lqtrans{A}'
	\ \dfn  \ 
	\begin{cases}
	\Prov{\Sigma^f_{k}}_{|A|}'(A) & \text{if }k=\ndcomp' (A)\leq \condcomp' (A) \\
	\Prov{\Pi^f_{k}}_{|A|}'(A) & \text{if }k=\condcomp'(A) < \ndcomp'(A)
	\end{cases}
	\]
\end{definition}

\noindent
We now have the following analogues of the main results of the previous section, proved by essentially the same arguments:

\begin{theorem}
	\label{thm:main2}
	We have the following:
	\begin{enumerate}
		\item\label{item:qltrans-nonaffine-correct} $\qltrans\cdot'$ is a polynomial-time encoding from $\qcpl$ to $\MALL$.
		\item\label{item:lqtrans-nonaffine-correct} $\lqtrans{\cdot}'$ is a polynomial-time encoding from $\MALL$ to $\qcpl$.
		\item\label{item:inverse-encodings-nonaffine} The composition $ \lqtrans{\cdot}' \circ \qltrans{\cdot}' : \qcpl \to \qcpl$ preserves quantifier complexity, i.e.\ for $k\geq 1$, it maps true $\Sigma^q_k$ ($\Pi^q_k$) sentences  to true $\Sigma^q_k$ (resp.~$\Pi^q_k$) sentences.
	\end{enumerate}
	
\end{theorem}

\begin{corollary}
	We have the following, for $k\geq 1$:
	\begin{enumerate}
		\item $\{ A : \text{$\ndcomp' (A) \leq k$ and $\MALL$ proves $A$} \}$ is $\Sigma^p_k$-complete.
		\item $\{ A  : \text{$\condcomp' (A) \leq k$ and $\MALL$ proves $A$}  \}$ is $\Pi^p_k$-complete.
	\end{enumerate}
\end{corollary}

\section{Conclusions and further remarks}
\label{sect:concs}

We gave a refined presentation of (multi-)focussed systems for multiplicative-additive linear logic, and its affine variant, that accounts for deterministic computations in proof search, cf.~Sect.~\ref{sect:prelims-ll-proof-search}.
We showed that it admits rather controlled normal forms in the form of \emph{bi-focussed} proofs, and highlighted a duality between focussing and `co-focussing' that emerges thanks to this presentation.
The main reason for using focussed systems such as ours was to better reflect the alternating time complexity of bottom-up proof search, cf.~Sect.~\ref{sect:proof-search-predicates}.
We justified the accuracy of these bounds by showing that natural measures of proof search complexity for $\afocMALL$ tightly delineate the theorems of $\aMALL$ according to associated levels of the polynomial hierarchy, cf.~Sects.~\ref{sect:qltrans} and \ref{sect:lqtrans}.
We were also able to obtain a similar delineation for $\MALL$ too, cf.~Sect.~\ref{sect:mallph}.
These results exemplify how the capacity of proof search to provide optimal decision procedures for logics extends to important subclasses of $\pspace$.
As far as we know, this is the first time such an investigation has been carried out.

\smallskip

Our presentation of $\focMALL(\mathsf w)$ should extend to logics with quantifiers and exponentials, following traditional approaches to focussed linear logic, cf.~\cite{Andreoli92,laurent:tel-00007884}.
It would be interesting to see what could be said about the complexity of proof search for such logics.
For instance, the usual $\forall$ rule becomes deterministic in our analysis, since it does not branch:
\[
\vlinf{\forall}{\text{$y$ is fresh}}{\Gamma, \forall x. A(x)}{\Gamma, A(y)}
\]
As a result, the alternation complexity of proof search is not affected by the $\forall$-rule, but rather interactions between positive connectives, including $\exists$, and negative connectives such as $\vlan$.
Interpreting this over a classical setting could even give us new ways to delineate true QBFs according to the polynomial hierarchy, determined by the alternation of $\exists$ and propositional connectives rather than $\forall$.
One issue here is that witnessing $\exists$ steps seems to significantly impact the complexity proof search.
Nonetheless this would be an interesting line of future research.

\smallskip

Much of the literature on \emph{logical frameworks} via focussed systems is based around the idea that an inference rule may be simulated by a `bi-pole', i.e.\ a single alternation between an invertible and non-invertible phase of inference steps.
However accounting for determinism might yield more refined simulations where, say, non-invertible rules are simulated by phases of deterministic and nondeterministic rules, but not co-nondeterministic ones.
In particular this should be possible for standard translations between modal logic and first-order logic, cf.~\cite{MillerV15,MarinMV16}.

\subsection*{Acknowledgements}
	I would like to thank Taus Brock-Nannestad, Kaustuv Chaudhuri, Sonia Marin and Dale Miller for many fruitful discussions about focussing, in particular on the presentation of it herein.
	I am particularly grateful to Sonia Marin and the anonymous reviewers of both this work and the previous conference version, \cite{Das18:ijcar}, for their careful and helpful proofreading.

\bibliographystyle{alpha}      
	\bibliography{focus-complex-bib}

\end{document}